\documentclass[10pt]{article}

\usepackage[margin=1.1in]{geometry}
\usepackage{amsmath,amssymb,amsfonts,amsthm}
\usepackage{mathtools}
\usepackage{graphicx}
\usepackage{xcolor}
\usepackage[authoryear, round]{natbib}
\usepackage{booktabs}
\usepackage{array}
\usepackage{enumitem}
\usepackage{bm}
\usepackage{titlesec}
\usepackage{fancyhdr}
\usepackage{microtype}
\usepackage{parskip}
\usepackage{float}
\usepackage{algorithm}
\usepackage[noend]{algpseudocode}

\usepackage{hyperref}
\usepackage[capitalize,nameinlink]{cleveref}
\usepackage{tikz}
\usetikzlibrary{arrows.meta,positioning,calc,fit,backgrounds}

\newtheorem{theorem}{Theorem}
\newtheorem{proposition}{Proposition}
\newtheorem{corollary}{Corollary}

\newtheorem{remark}{Remark}
\newtheorem{assumption}{Assumption}

\crefname{algorithm}{Algorithm}{Algorithms}
\crefname{figure}{Figure}{Figures}
\Crefname{figure}{Figure}{Figures}

\hypersetup{
  colorlinks=true,
  linkcolor=blue!70!black,
  citecolor=blue!50!black,
  urlcolor=blue!60!black
}

\providecommand{\E}{\mathbb{E}}
\providecommand{\Prob}{\mathbb{P}}

\newcommand{\one}{\mathbf 1}

\newcommand{\cG}{\mathcal G}

\newcommand{\cX}{\mathcal X}
\newcommand{\cY}{\mathcal Y}

\newcommand{\Dtr}{\mathcal D_{\mathrm{tr}}}
\newcommand{\Dcal}{\mathcal D_{\mathrm{cal}}}
\newcommand{\Dsrc}{\mathcal D_{\mathrm{src}}}

\newcommand{\defeq}{\coloneqq}
\newcommand{\KL}{D_{\mathrm{KL}}}

\newcommand{\ess}{\mathrm{ESS}}

\newcommand{\cali}{\mathrm{cal}}
\newcommand{\conf}{\mathrm{conf}}

\titleformat{\section}{\large\bfseries\color{blue!60!black}}{{\thesection}}{1em}{}[\titlerule]
\titleformat{\subsection}{\normalsize\bfseries\color{blue!40!black}}{{\thesubsection}}{1em}{}
\titleformat{\subsubsection}{\normalsize\bfseries\color{blue!40!black}}{{\thesubsubsection}}{1em}{}

\begin{document}

\begin{center}
  {\LARGE\bfseries \textsf{Anytime-Valid Confirmation of Covariate Balance \\[5pt]
for Prespecified Corrections}}\\[2em]
  {\large Seungjin Choi}\\[0.5em]
  {\normalsize CROID Research and aSSIST University, Seoul, Korea}
\end{center}

\vspace{0.5em}
\noindent\rule{\linewidth}{1.5pt}
\vspace{0.5em}

\begin{abstract}
Many covariate-shift adaptation methods construct a correction $w(x)$, but users must still determine whether the corrected distribution is sufficiently balanced for the target stream. We study anytime-valid confirmation of prespecified corrections from sequential unlabeled target inputs. Our primary contribution is a procedure for confirming covariate balance. For a prespecified class of balancing functions and tolerances, time-uniform confidence sequences permit continuous monitoring and data-dependent stopping once all plausible target moments lie within their tolerance bands. If the correction is out of tolerance for at least one function, the probability of ever incorrectly confirming balance is at most the prescribed level. Upon stopping, the procedure yields a certificate local to the chosen functions and tolerances, yet providing an absolute downstream-adequacy statement that ordinary shift diagnostics generally do not. With finite source data, contracted bands preserve this guarantee while accounting for uncertainty in weighted source moments, whereas expanded bands support only compatibility diagnostics. As complementary information, we study a source-calibrated likelihood-ratio e-process whose KL-drift identity characterizes correction directions relative to the source. Under the source-reference distribution, the probability of ever crossing its evidence threshold is controlled, but crossing does not confirm balance. We also give an exponential-tilt test for departures beyond an acceptable correction region and deploy balance-confirmed corrections in weighted conformal prediction. Experiments illustrate false-confirmation control, locality to the balancing-function class, KL-drift diagnostics, acceptable-region monitoring, finite-source effects, and downstream conformal coverage under covariate shift.
\end{abstract}


\section{Introduction}
\label{sec:intro}

Prediction and decision systems are often deployed in environments whose inputs
no longer look like the data used for training or calibration.  A model may be
trained on historical customers and deployed on a new market, calibrated on
passive observations and used after an active acquisition policy, or tuned on
simulated inputs and applied to real measurements.  In such settings, the
conditional response mechanism may remain stable while the input distribution
changes.  This is the pure covariate-shift model, in which
\(P_t^X\neq P_s^X\) while \(P_t(Y\mid X)=P_s(Y\mid X)\), where the superscripts
denote input marginals of the source and target distributions.
Covariate-shift methods therefore try to repair the mismatch by
reweighting source inputs.  The ideal target-to-source correction is the
Radon--Nikodym derivative
\begin{equation}
  w^\star(x)=\frac{dP_t^X}{dP_s^X}(x),
  \label{eq:true-density-ratio}
\end{equation}
which, when densities exist, is the familiar covariate density ratio.  This idea
underlies importance-weighted empirical risk minimization
\citep{ShimodairaH2000jspi,CortesC2010neurips}, direct density-ratio estimation
\citep{SugiyamaM2007neurips,KanamoriT2009jmlr}, and covariate-weighted conformal
prediction \citep{TibshiraniR2019neurips}.

Constructing a correction, however, is not the same as knowing that it is safe to
use.  A density-ratio estimate may be noisy, a simulator may be biased, a domain
classifier may exploit irrelevant artifacts, and a correction that worked during
a previous deployment may fail as the target population drifts.  Downstream users
need an inferential answer to a post-construction question: after a candidate
correction has been fixed, is the corrected source distribution sufficiently
credible for the target stream and sufficiently balanced for the intended use?
This question is especially natural when target inputs arrive sequentially and
unlabeled.  One may monitor the stream for days, weeks, or until enough evidence
has accumulated, so a fixed-time diagnostic is not enough.  The evidence must
remain valid under optional stopping.

In this paper we study that post-construction problem.  A practitioner supplies a
nonnegative, prespecified covariate correction \(w\).  We write \(\E_s\),
\(\E_t\), and \(\E_w\) for expectations under the source, target, and
corrected input distributions, respectively.  We view \(w\) as the
Radon--Nikodym derivative of the corrected input distribution with respect to
the source input distribution, inducing
\begin{equation}
  dP_w^X(x)=w(x)\,dP_s^X(x),
  \qquad \E_s \left[ w(X) \right]=1,
  \label{eq:corrected-distribution-intro}
\end{equation}
with finite-source normalization handled separately when \(w\) is estimated from
data.  The correction is the input to our procedure.  The inferential target is
whether the induced corrected input distribution \(P_w^X\) is supported by the
sequential target inputs and balanced enough for downstream use.

There is no single scalar assessment that resolves this question.  The paper's
primary inferential target is local but absolute: on the prespecified functions
that matter for the downstream task, is the remaining mismatch within tolerance?
A complementary question is global but directional: does a source-calibrated
likelihood ratio favor the corrected input distribution over the unweighted
source on the incoming stream?  Neither answer implies the other.  A correction
may improve the overall input distribution while still missing an important
downstream feature.  Conversely, a correction may balance a narrow collection of
features while remaining poor elsewhere.  We therefore distinguish the primary anytime-valid balance-confirmation
procedure from a complementary global monitor.  Both use the same target input
stream, but only successful balance confirmation yields the certificate that
supports the absolute downstream-adequacy claim named in the title.

\begin{itemize}
\item \textbf{Primary anytime-valid covariate balance confirmation.}  Let
\(\mathcal F\) be a prespecified class of balancing functions, and define
\[
  \Delta_{\mathcal F}(w)
  =
  \sup_{f\in\mathcal F}
  \left|\E_t \left[ f(X) \right]-\E_w \left[ f(X) \right]\right|.
\]
Given prespecified tolerances, time-uniform confidence sequences allow continuous
monitoring and stopping as soon as every plausible target balancing-function
mean lies within its tolerance band.  If the proposed correction is out of
tolerance for at least one balancing function, the probability of ever falsely
confirming balance is at most the prescribed level.  Upon stopping, the procedure yields a certificate local to
\(\mathcal F\), with the explicit tolerance-level guarantee needed for
downstream deployment.

\item \textbf{Supporting source-calibrated global monitor.}  The product
\[
  M_n=\prod_{i=1}^n w(X_i)
\]
is an e-process when the source input distribution is used as the calibration
reference.  Crossing its boundary is therefore anytime-valid under a source-like
stream.  Its log-growth has a KL-drift interpretation under a general monitoring
distribution: it grows when the corrected input distribution is closer than the
source in KL risk, and it tends to decay when the correction is globally worse.
This monitor is not restricted to the chosen balancing-function class and can
warn about globally harmful correction directions, but it is only relative.  It
need not reveal an incomplete correction that still improves on the source in KL
risk, does not provide a uniformly level-controlled test of the composite claim
that the correction is no better than the source, and cannot confirm that the residual mismatch is small enough for a specific
downstream use.
\end{itemize}

The two assessments provide nonredundant but non-exhaustive information rather
than forming a pipeline.  The global monitor can warn that a correction is
broadly harmful even when a poorly chosen balancing-function class would pass
it.  Their blind spots can nevertheless overlap: an incomplete correction may
satisfy the selected balancing functions and still improve on the source in KL
risk, producing favorable global drift.  Covariate balance confirmation provides
the explicit downstream adequacy claim that the global monitor alone cannot
provide.  The balancing-function class specifies which moments, regions,
representations, or score-relevant summaries must be balanced, and the tolerance
specifies how much residual mismatch is acceptable.

The comparison with label shift helps explain why the monitored object is the
input distribution.  Under label shift, \(P_s(X\mid Y)=P_t(X\mid Y)\), while the
predictive distribution \(P_t(Y\mid X)\) generally changes; a correction can then
be represented as a predictive tilt, as in \citet{Choi2026testing}.  Under pure
covariate shift, the conditional response mechanism is invariant, so the
correction should act primarily on the input marginal rather than on
\(Y\mid X\).  \Cref{tab:shift-correction-objects} summarizes the distinction.

\begin{table}[H]
\centering
\caption{Shift regimes and corresponding correction objects.  Label shift
naturally leads to label weights or predictive tilting, whereas pure covariate
shift calls for an input-distribution correction.  General joint shift requires
additional structure or robustness assumptions.}
\label{tab:shift-correction-objects}
\renewcommand{\arraystretch}{1.14}
\resizebox{\linewidth}{!}{%
\begin{tabular}{@{}llll@{}}
\toprule
\textbf{Shift type} & \textbf{Invariant assumption} & \textbf{Correction object} &
\textbf{Main repair} \\
\midrule
Label shift &
\(P_s(X\mid Y)=P_t(X\mid Y)\) &
\(w_Y(y)=dP_t^Y/dP_s^Y\) &
Predictive tilt; weighted calibration \\
Covariate shift &
\(P_s(Y\mid X)=P_t(Y\mid X)\) &
\(w_X(x)=dP_t^X/dP_s^X\) &
Weighted calibration; optional weighted fitting \\
Joint shift &
No generic invariant &
Structured \(w(x,y)\) or \(dP_t^{X,Y}/dP_s^{X,Y}\) &
Joint weighting; robustness \\
\bottomrule
\end{tabular}%
}
\end{table}

Because the global monitor and the balance-confirmation procedure are both
input-facing, unlabeled target inputs are enough for monitoring.  Labels may be needed later for training, calibration, or
evaluation, but they are not needed to confirm covariate balance itself.
\Cref{fig:method-schematic} shows the resulting structure: a fixed correction and
a target input stream feed two parallel anytime-valid assessments, which then
inform downstream use.

\begin{figure}[ht!]
\centering
\footnotesize
\begin{tikzpicture}[
  >=Latex,
  fixedbox/.style={
    draw,
    rounded corners=5pt,
    thick,
    align=center,
    text width=3.15cm,
    minimum width=3.45cm,
    minimum height=1.95cm,
    inner sep=4pt
  },
  arr/.style={->,thick}
]

\node[fixedbox] (input) at (0,0) {
  \textbf{Inputs}\\[2pt]
  source data \(\mathcal D_{\mathrm{src}}\), fixed correction \(w(x)\),
  and target inputs \(X_1,X_2,\ldots\)
};

\node[fixedbox] (relative) at (4.3,1.55) {
  \textbf{Supporting}\\
  \textbf{global monitor}\\[3pt]
  \(M_n=\prod_{i=1}^n w(X_i)\)\\[1pt]
  global, directional
};

\node[fixedbox] (balance) at (4.3,-1.55) {
  \textbf{Primary balance}\\
  \textbf{confirmation}\\[3pt]
  \(\Delta_{\mathcal F}(w)\le \varepsilon\)\\[1pt]
  \(\Prob\left\{\tau_{\mathrm{bal}}<\infty\right\}\le\delta\)
  under imbalance
};

\node[
  fixedbox,
  text width=3.65cm,
  minimum width=4.05cm
] (finite) at (8.6,1.55) {
  \textbf{Finite-source uncertainty}\\[2pt]
  normalizer uncertainty affects the global monitor; weighted source-moment
  uncertainty contracts balance bands
};

\node[fixedbox] (down) at (8.6,-1.55) {
  \textbf{Downstream task}\\[2pt]
  weighted conformal calibration under pure covariate shift
};

\draw[arr] (input) -- (relative);
\draw[arr] (input) -- (balance);
\draw[arr] (balance) -- (down);

\draw[arr,dashed] (finite.west) -- (relative.east);

\draw[arr,dashed,rounded corners=4pt]
  (finite.south)
  -- ++(0,-0.25)
  -| (balance.north);

\end{tikzpicture}

\caption{
The primary anytime-valid balance-confirmation procedure and a supporting global
monitor, run in parallel on the target input stream. Covariate balance
confirmation controls the probability of ever falsely confirming an
out-of-tolerance correction and is the gate for downstream use. The global
e-process provides source-calibrated relative evidence and a KL-drift diagnostic
without restricting attention to the selected balancing functions, but it does
not itself confirm balance with the target.
}
\label{fig:method-schematic}
\end{figure}
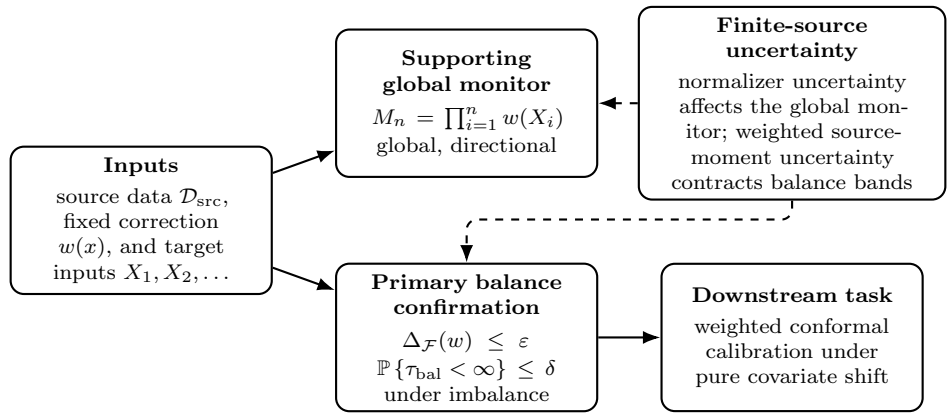

The main contributions are as follows.
\begin{itemize}
\item Our primary result is an anytime-valid procedure for confirming covariate balance for a
prespecified correction.  For a fixed balancing-function class and tolerances,
the procedure may be monitored continuously and stopped at any data-dependent
time; if the correction is out of tolerance for at least one balancing function,
the probability of ever falsely confirming balance is at most \(\delta\).

\item We show that this optional-stopping guarantee follows from simultaneous
time-uniform coverage of the target balancing-function means.  The confirmation guarantee therefore need not be represented by an e-process.  We provide an explicit
Hoeffding construction for bounded functions, a sub-Gaussian construction for
the Gaussian experiments, and references to variance-adaptive alternatives.

\item We separate the absolute, balancing-function-local confirmation procedure from a
source-calibrated global e-process.  The global process is distribution-wide and
has a KL-drift interpretation, but it provides relative evidence rather than a
level-controlled confirmation of balance with the target.

\item We analyze finite-source uncertainty in normalization and weighted source
moments.  Formal confirmation requires contracted confirmation bands, yielding a
\(\delta+\eta\) false-confirmation bound when source moments are estimated,
whereas expanded compatibility bands support only diagnostic compatibility
claims.

\item We develop two operational uses of the framework: anytime-valid testing
of whether a target stream exceeds a prespecified acceptable exponential-tilt
region, and deployment of a balance-confirmed correction in weighted conformal
prediction.  Experiments evaluate these uses together with the core global and
balance procedures and finite-source effects.
\end{itemize}

\section{Related Work}
\label{sec:related}

\textbf{Conformal prediction under covariate shift.}
Conformal prediction gives finite-sample predictive coverage under exchangeability
\citep{VovkV2005book,PapadopoulosH2002ecml,ShaferG2008jmlr,BarberRF2023aos,AngelopoulosAN2023ftml}.  
When the deployment input marginal differs from the calibration input marginal, exchangeability is broken.
Weighted conformal prediction repairs this mismatch by replacing the ordinary
calibration distribution with the target-weighted one; with the correct
covariate density ratio, target coverage can be recovered
\citep{TibshiraniR2019neurips}.  Our paper is not a new weighted conformal prediction method.  It addresses a
question that weighted conformal procedures usually take for granted: before
using a proposed correction in calibration, can we sequentially confirm that the
induced covariate balance is adequate for the score distribution or other
features relevant to downstream validity?

\textbf{Covariate shift and density-ratio estimation.}
A large literature treats covariate-shift correction as an estimation problem.
The target-to-source ratio \(w^\star(x)=dP_t^X/dP_s^X(x)\) is estimated from
source and target inputs and then used for importance-weighted risk minimization,
weighted calibration, or diagnostics.  Direct density-ratio methods estimate the
ratio without separately estimating the two input densities, including KLIEP,
least-squares importance fitting, and related classifier-odds approaches
\citep{SugiyamaM2007neurips,KanamoriT2009jmlr,SugiyamaM2012book}.  This paper is
complementary.  We assume a proposed correction \(w\) has already been fixed
relative to the monitoring stream, possibly from older data, a simulator, a
known policy change, or a separate domain classifier.  The goal is not to learn \(w^\star\) from the monitored inputs.  The proposed
correction is treated as fixed, and the inferential target is the balance it
induces, assessed in two distinct senses: globally relative to the
unweighted source, and, in the tolerance sense that matters downstream, relative
to the target within a prespecified balancing-function class.

\textbf{Covariate balancing weights.}
Existing covariate-balancing methods focus primarily on constructing weights
that satisfy balance constraints on a finite sample.  Entropy balancing
\citep{HainmuellerJ2012pa}, the covariate balancing propensity score
\citep{ImaiK2014jrsssb}, stable balancing weights \citep{ZubizarretaJR2015jasa},
empirical balancing calibration weighting \citep{ChanKCG2016jrsssb}, and
approximate residual balancing \citep{AtheyS2018jrsssb} choose weights so that
prescribed balancing functions match across treatment groups or between a
weighted sample and a target population; see \citet{Ben-MichaelE2021arxiv} for
a review.  Our objective is different.  We assume that a proposed correction
\(w\) has already been constructed and fixed before monitoring, and develop
anytime-valid sequential procedures for confirming whether the corrected source
distribution achieves a prescribed balance criterion on an incoming target
stream.  Thus, the causal-inference literature primarily addresses a
\emph{construction problem}, whereas our work addresses a
\emph{confirmation problem}.  The two roles are complementary: a
covariate-balancing procedure may supply \(w\), after which our method can assess
over time whether the resulting corrected distribution remains adequately
balanced for downstream use.

\textbf{Two-sample testing, integral probability metrics, and equivalence.}
The balancing discrepancy
\(\Delta_{\mathcal F}(w)=\sup_{f\in\mathcal F}|\E_t \left[ f(X) \right]-\E_w \left[ f(X) \right]|\)
is an integral probability metric \citep{MullerA1997aap,SriperumbudurBK2012ejs},
of which the kernel maximum mean discrepancy is the reproducing-kernel special
case \citep{GrettonA2012jmlr}.  Classical two-sample tests, including sequential
ones, ask whether such a discrepancy is nonzero; recent anytime-valid two-sample
tests do this by betting on witness functions drawn from the variational form of
an integral probability metric \citep{ShekharS2024ieeetit}, with related
martingale methods for sequentially estimating divergences
\citep{ManoleT2023ieeetit}.  Our inferential direction is reversed.  Rather than detect a difference, we seek anytime-valid confirmation that the residual
discrepancy is small enough,
\[
  \sup_{f\in\mathcal F}
  \left|\E_t \left[ f(X) \right]-\E_w \left[ f(X) \right]\right|
  \le \varepsilon ,
\]
an equivalence-style statement in the spirit of the two one-sided tests framework
\citep{SchuirmannDJ1987jpb,WellekS2010book}, implemented sequentially with
time-uniform confidence sequences.  The balancing-function class \(\mathcal F\) makes the resulting certificate
deliberately local to the chosen notion of balance: raw covariates,
model features, subpopulation indicators, kernel landmarks, or other prespecified
diagnostics.

\textbf{E-values, anytime-valid inference, and confidence sequences.}
Nonnegative supermartingales and e-values provide optional-stopping guarantees
through Ville's inequality \citep{VilleJ1939phd}.  Sequential likelihood ratios
connect the same idea to the sequential probability ratio test
\citep{WaldA1945aoms}.  The modern safe-testing and game-theoretic statistics
literature develops e-values and betting interpretations of sequential evidence
\citep{GrunwaldP2024jrsssb,RamdasA2023ss}, while confidence sequences provide
practical time-uniform inference for means and bounded observations
\citep{HowardSR2021aos,Waudby-SmithI2024jrsssb}.  We use these tools in two
complementary roles.  The input e-process gives a global, directional monitor
for the proposed ratio, while the confidence-sequence construction gives a local,
tolerance-level route to covariate balance confirmation.

\textbf{Relation to label-shift confirmation.}
The closest conceptual predecessor is anytime-valid confirmation of label-shift
corrections \citep{Choi2026testing}.  In that setting, the correction tilts the
predictive distribution \(p_s(y\mid x)\), and evidence is accumulated from
labeled target pairs.  Under pure covariate shift, the conditional predictive
distribution is invariant, so the correction acts on the input marginal instead.
This changes both the data stream and the meaning of confirmation: target inputs alone can generate source-calibrated relative evidence, but covariate balance confirmation is needed to establish that the corrected input distribution is close enough for downstream use.

\section{Problem Setup}
\label{sec:setup}

Let \(P_s\) and \(P_t\) denote source and target joint distributions on
\(\cX\times\cY\).  Their input marginals are \(P_s^X\) and \(P_t^X\).  The pure
covariate-shift model assumes
\begin{equation}
  P_t^X\neq P_s^X,
  \qquad
  P_t(Y\mid X=x)=P_s(Y\mid X=x),
  \label{eq:cov-shift-setup}
\end{equation}
so the shift is in the distribution of inputs, not in the conditional response
mechanism.  Source labeled data are available,
\[
  \Dsrc=\{(X_i^s,Y_i^s)\}_{i=1}^{n_s},
  \qquad
  (X_i^s,Y_i^s)\sim P_s,
\]
and, for the main development, target inputs arrive sequentially as
\[
  X_1,X_2,\ldots \stackrel{\mathrm{iid}}{\sim}P_t^X.
\]
Throughout, \(X_i\) denotes the \(i\)th target input.  We reserve \(t\) for the
target distribution and use \(n\) for the current monitoring index.  Target
responses are not needed to confirm covariate balance.  They may be used later
for downstream training, calibration, or evaluation, but the confirmation
procedures in this paper are driven by the target input stream.  The abstract
confidence-sequence result extends beyond iid sampling whenever a valid
time-uniform confidence sequence is available for a fixed target parameter under
the actual sequential distribution.

A practitioner proposes a nonnegative covariate correction
\[
  w:\cX\to[0,\infty).
\]
The correction is prespecified relative to the monitoring stream.  It may have
been obtained from source data, older target inputs, a held-out source-target
classifier, a simulator, a known sampling-policy change, a transport map, or
domain knowledge.  For the main results, all information used to construct
\(w\) is treated as part of the initial sigma-field.  Conditional on this
information, \(w\) is fixed before the monitored inputs
\(X_1,X_2,\ldots\) are used to form the evidence process.  The paper focuses
on this prespecified-correction setting.  Online updating is possible in
principle, but it would require a separate predictable construction: the weight
\(w_i\) used at monitoring index \(i\) must be chosen from the past, before observing
\(X_i\), and must remain normalized under the source input distribution.  We
do not develop this adaptive extension here.

Recall that \(\E_s\), \(\E_t\), and \(\E_w\) denote expectations under
\(P_s^X\), \(P_t^X\), and the corrected input distribution \(P_w^X\),
respectively.  Let
\(\cG_0\) contain all information used before monitoring, including the source
data, the construction and normalization of \(w\), and any source-moment
intervals.  The monitoring filtration is
\[
  \cG_n
  =
  \cG_0\vee\sigma(X_1,\ldots,X_n),
  \qquad n\ge 0.
\]
Thus every quantity constructed before the target stream is observed is
\(\cG_0\)-measurable.

\subsection{Corrected input distribution}
\label{subsec:corrected-input-distribution}

For any proposed correction satisfying
\[
  0<Z_s(w)\defeq \E_s \left[ w(X) \right]<\infty,
\]
define the normalized correction
\[
 \overline{w}(x)=\frac{w(x)}{Z_s(w)}
\]
and the corrected input distribution
\begin{equation}
  P_w^X(A)
  =
  \E_s\!\left[\overline w(X)\one\{X\in A\}\right]
  =
  \frac{\E_s \left[w(X) \, \one\{X\in A\} \right]}{\E_s \left[w(X)\right]}.
  \label{eq:corrected-input-distribution}
\end{equation}
Thus \(dP_w^X/dP_s^X=\overline w\).  If the proposed correction is already normalized,
\(\E_s \left[ w(X) \right]=1\), then \(\overline w=w\).  
We assume this normalized case to keep notation simple for the main development.
Finite-source normalization and conservative normalizers are treated in \cref{sec:finite-source}.

When the target-to-source density ratio exists,
\[
  w^\star(x)=\frac{dP_t^X}{dP_s^X}(x),
\]
and \(w=w^\star\), the corrected input distribution satisfies \(P_w^X=P_t^X\).
For an approximate correction, \(P_w^X\) is the target input distribution implied by the proposed weights.  
The goal of this paper is to confirm, sequentially and with optional-stopping validity, whether that implied input distribution achieves a prespecified level of covariate balance with the target.

\subsection{Primary confirmation procedure and supporting monitor}
\label{subsec:two-tools}

The framework separates a global directional comparison from the
downstream-facing balance target.  The first question is, ``Does a likelihood
ratio calibrated under the source input distribution favor the reweighted source
distribution \(P_w^X\) over \(P_s^X\) on the incoming stream?''  This is organized
around the source-reference calibration null
\begin{equation}
  H_0^{\mathrm{src}}:\quad
  X_i\mid \cG_{i-1}\sim P_s^X,
  \quad i=1,2,\ldots.
  \label{eq:source-input-null}
\end{equation}
Under the pure covariate-shift premise \eqref{eq:cov-shift-setup} we already expect
\(P_t^X\neq P_s^X\), so \(H_0^{\mathrm{src}}\) is not a hypothesis we believe.  It
is a calibration device: crossing the input e-process of \cref{sec:eprocess}
controls the crossing probability when the monitored stream is source-like.  The
process drift under a general monitoring distribution is positive when
\(P_w^X\) is closer than \(P_s^X\) in KL risk and negative when it is farther away
(\cref{prop:kl-drift}).  This drift interpretation is scientifically useful, but
the source-calibrated crossing guarantee is not a level-controlled test of the
composite null that the correction is no better than the source.  The global
monitor therefore supplies relative evidence and diagnostics; it does not by
itself establish covariate balance with the target.

The second question is absolute up to a chosen tolerance.  Given a balancing-function class
\(\mathcal F\) and tolerance \(\varepsilon\ge 0\), define the residual balancing
imbalance
\begin{equation}
  \Delta_{\mathcal F}(w)
  =
  \sup_{f\in\mathcal F}
  \big|
  \E_t \left[ f(X) \right]-\E_w \left[ f(X) \right]
  \big|,
  \qquad
  \E_w \left[ f(X) \right]
  =
  \frac{\E_s  \left[ w(X)f(X) \right]}{\E_s \left[ w(X)\right]}.
  \label{eq:ipm-discrepancy}
\end{equation}
The associated beyond-tolerance imbalance null is
\begin{equation}
  H_0^{\mathrm{bad}}(\varepsilon):
  \quad
  \Delta_{\mathcal F}(w)>\varepsilon.
  \label{eq:bad-null}
\end{equation}
Rejecting \eqref{eq:bad-null} gives covariate balance confirmation: the
remaining discrepancy between \(P_w^X\) and \(P_t^X\), as measured by the chosen
balancing-function class, is within the prespecified tolerance.  The balancing-function class should be
chosen before monitoring begins and should reflect the covariate features or
moments that matter for the intended downstream use.  The confidence-sequence
constructions in \cref{sec:balance} impose the corresponding boundedness or
sub-Gaussian conditions on these balancing functions.
The null in \eqref{eq:bad-null} uses a strict inequality.  Hence the boundary
case \(\Delta_{\mathcal F}(w)=\varepsilon\) is not counted as a false confirmation
of \(\varepsilon\)-balance.  A small margin can be introduced if closed-null
boundary control is desired.

The two assessments have different targets and neither implies the other.
The source-calibrated global monitor compares \(P_w^X\) with \(P_s^X\) through
\(\log w\).  It is not restricted to the balancing-function class, but it
returns only a relative direction of improvement and can favor an incomplete
correction that remains far from \(P_t^X\).  Covariate balance confirmation
compares \(P_w^X\) with \(P_t^X\) through the balancing-function class
\(\mathcal F\); it is a within-tolerance equivalence question that returns a
quantitative guarantee only on the chosen balancing functions.  The two
assessments are therefore run in parallel and reported together because they
provide different information, not because either eliminates all blind spots of
the other.

In this paper we present the supporting global monitor first because its simple
likelihood-ratio construction makes the limitations of relative evidence
transparent.  Section~\ref{sec:eprocess} develops this monitor and its KL-drift
interpretation.  Section~\ref{sec:balance} then develops the paper's primary
result, an anytime-valid procedure for confirming covariate balance.  Thus the order of
presentation is pedagogical rather than a ranking of the contributions.

\section{Source-Calibrated Global Relative Evidence from an Input E-Process}
\label{sec:eprocess}

We begin with the simpler supporting global monitor because it exposes what
source-calibrated relative evidence can and cannot establish.  The process
compares the corrected input distribution \(P_w^X\) with the source reference
\(P_s^X\) over the full input space.  Its boundary-crossing guarantee is calibrated
under the source-reference law.  Under the actual monitoring distribution, its
log-growth provides a diagnostic of whether the proposed correction points in a
globally favorable or harmful KL direction relative to leaving the source
unweighted.  This monitor is deliberately coarse and directional.  In particular,
it does not provide the absolute balance claim required for downstream use.  Its
limitations motivate the primary anytime-valid balance-confirmation procedure developed in
\cref{sec:balance}.

\bigskip
\begin{assumption}[Prespecified normalized correction]
\label{ass:prespecified}
The function \(w:\cX\to[0,\infty)\) is measurable and
\(\cG_0\)-measurable, is fixed before the confirmation stream is used, and
satisfies
\[
  \E_s \left[ w(X) \right]=1
\]
almost surely with respect to any pre-monitoring randomness in \(\cG_0\).
Equivalently,
\[
  dP_w^X(x)=w(x)\,dP_s^X(x),
\]
so \(w=dP_w^X/dP_s^X\) is the Radon--Nikodym derivative of the corrected input
distribution with respect to the source input distribution.
\end{assumption}

\bigskip
\begin{theorem}[Input likelihood-ratio e-process]
\label{thm:input-eprocess}
Under \cref{ass:prespecified} and the source-reference calibration null
\eqref{eq:source-input-null},
\[
  e_i=w(X_i)
\]
is a conditional e-value:
\[
  \E_{H_0^{\mathrm{src}}}\!\left[e_i\mid \cG_{i-1}\right]=1.
\]
Consequently,
\begin{equation}
  M_n=\prod_{i=1}^n w(X_i),
  \qquad M_0=1,
  \label{eq:input-martingale}
\end{equation}
is a nonnegative martingale under the source-reference calibration null.  For any
\(\alpha\in(0,1)\),
\begin{equation}
  \Prob_{H_0^{\mathrm{src}}}
  \left\{
  \sup_{n\ge 0}M_n\ge \frac{1}{\alpha}
  \right\}
  \le \alpha.
  \label{eq:ville-input}
\end{equation}
\end{theorem}

\noindent\emph{Proof sketch.}
Under the source-reference calibration null, the increment has conditional mean
one after conditioning on the past, so the product is a nonnegative martingale.
Ville's inequality gives the time-uniform crossing bound.  The full proof is in
Appendix~\ref{app:proof-input-eprocess}.

The stopping time
\[
  \tau_{\conf}
  =
  \inf\left\{n\ge 1:M_n\ge \frac{1}{\alpha}\right\}
\]
is an anytime-valid source-calibrated crossing rule.  Its level guarantee is with
respect to the source-reference calibration null.

The log process is
\[
  \log M_n
  =
  \sum_{i=1}^n \log\frac{dP_w^X}{dP_s^X}(X_i).
\]
It is the cumulative log-likelihood advantage of the corrected input
distribution over the source input distribution on the observed stream.
Therefore the e-process is not merely a test statistic.  It is an online
accounting device whose average log-growth can be interpreted through the KL
difference in \cref{prop:kl-drift}.

\bigskip
\begin{proposition}[KL drift identity]
\label{prop:kl-drift}
Let \(Q^X\) be any input distribution satisfying
\[
  Q^X\ll P_w^X\ll P_s^X,
  \qquad
  \E_Q \left[ |\log w(X)| \right]<\infty,
  \qquad
  \KL(Q^X\|P_s^X)<\infty.
\]
If \(X_i\) has marginal distribution \(Q^X\), then the expected one-step log
growth of \eqref{eq:input-martingale} is
\begin{equation}
  \E_Q \left[ \log w(X) \right]
  =
  \KL(Q^X\|P_s^X)-\KL(Q^X\|P_w^X).
  \label{eq:kl-drift}
\end{equation}
\end{proposition}

\noindent\emph{Proof sketch.}
The identity follows from the Radon--Nikodym chain rule
\[
  \frac{dQ^X}{dP_s^X}
  =
  \frac{dQ^X}{dP_w^X}\frac{dP_w^X}{dP_s^X},
\]
followed by taking \(Q^X\)-expectations.  The full proof is in
Appendix~\ref{app:proof-kl-drift}.

\cref{prop:kl-drift} gives the scientific interpretation of the global
relative monitor.  Although the proposition is stated for a general monitoring
distribution \(Q^X\), the intended case in deployment is \(Q^X=P_t^X\).  The
process grows when \(P_w^X\) is closer to the target input distribution than
\(P_s^X\) is in KL risk.  If \(P_w^X=P_t^X\), the drift is
\(\KL(P_t^X\|P_s^X)\).  If \(P_w^X\) improves over the source but is not exact,
the drift is still positive.  If \(P_w^X\) is worse than the source, the drift is
negative.

\bigskip
\begin{corollary}[Consistency under positive drift]
\label{cor:positive-drift}
Suppose \(X_i\) are iid from \(Q^X\) with \(Q^X\ll P_w^X\ll P_s^X\),
\(\E_Q \left[ |\log w(X)| \right]<\infty\), \(\KL(Q^X\|P_s^X)<\infty\), and
\[
  \KL(Q^X\|P_s^X)>\KL(Q^X\|P_w^X).
\]
Then
\[
  \frac{1}{n}\log M_n
  \to
  \KL(Q^X\|P_s^X)-\KL(Q^X\|P_w^X)
  \quad\text{almost surely},
\]
and \(\tau_{\conf}<\infty\) almost surely for every \(\alpha\in(0,1)\).
\end{corollary}

\noindent\emph{Proof sketch.}
Apply the strong law of large numbers to the increments \(\log w(X_i)\).  A
positive limiting average log growth forces \(M_n\) to cross every fixed
threshold.  The full proof is in Appendix~\ref{app:proof-positive-drift}.

The resulting online procedure is summarized in \cref{alg:better-source}.  It
multiplies the normalized correction values along the target input stream and
stops once the evidence exceeds the anytime-valid threshold.  A threshold
crossing is positive global relative evidence.  If the threshold is not crossed,
the appropriate conclusion is non-confirmation; by itself it is not a
level-controlled statement that the correction is worse than the source.  The
KL-drift identity still makes the path informative as a diagnostic.  Under
negative drift, \(n^{-1}\log M_n\) tends to a negative limit, so the process
typically decays rather than crossing the confirmation boundary.  This is one
failure mode that covariate balance confirmation can miss when the
balancing-function class is too narrow.  We do not turn decay into a formal
formal refutation procedure here, which would require a separate martingale under a
``no worse than source'' null.  We leave this as an open direction.

\begin{algorithm}[H]
\caption{Source-calibrated global relative evidence for a covariate correction}
\label{alg:better-source}
\begin{algorithmic}[1]
\Require Prespecified normalized correction \(w\) with \(\E_s \left[ w(X) \right]=1\), level \(\alpha\)
\State Set \(M_0=1\)
\For{\(n=1,2,\ldots\)}
  \State Observe target input \(X_n\)
  \State Compute \(e_n=w(X_n)\)
  \State Update \(M_n=M_{n-1}e_n\)
  \If{\(M_n\ge 1/\alpha\)}
    \State \Return positive global relative evidence for \(P_w^X\) over \(P_s^X\)
  \EndIf
\EndFor
\end{algorithmic}
\end{algorithm}

\section{Covariate Balance Confirmation and Finite-Source Effects}
\label{sec:balance}

The global e-process of \cref{sec:eprocess} answers a directional question:
whether the reweighted source input distribution is favored over the unweighted
source input distribution on the target stream.  It does not answer the absolute
question needed for downstream use.  A correction may be globally better than no
correction and still leave a residual covariate mismatch that is too large for
weighted calibration, weighted fitting, or diagnostics.  Covariate balance
confirmation therefore asks a within-tolerance question: is the corrected input
distribution close enough to the target input distribution on the prespecified
balancing functions?  Global relative evidence says whether the source-calibrated likelihood ratio
favors the correction over the unweighted source, whereas covariate balance
confirmation says whether it is good enough for the chosen downstream balance
criterion.  The latter is
successful confirmation that yields the tolerance-level certificate carrying the downstream guarantee, and it is
where the paper's main development lies.

This section separates three objects that are easy to conflate.  First, with known
weighted source moments, \(\varepsilon\)-balance is confirmed by checking whether
time-uniform target confidence sequences lie inside fixed tolerance bands.  Second,
with finite source samples, the weighted source moments themselves are uncertain,
so a stronger containment condition is needed to confirm the same
\(\varepsilon\)-balance claim.  Third, an expanded finite-source band is useful as
a compatibility diagnostic, but it supports only a weaker statement and should not be
reported as covariate balance confirmation.

\subsection{Balancing-function discrepancy}
\label{subsec:balancing-functions}

Let \(\mathcal F\) be a class of measurable real-valued functions
\[
  f:\cX\to\mathbb R
\]
for which the target and corrected expectations exist.  We call them
\emph{balancing functions}, following the covariate-balancing literature,
because their weighted moments define the aspects of the input distribution that
the correction is required to match.
For a proposed correction \(w\), define
\[
  \Delta_{\mathcal F}(w)
  =
  \sup_{f\in\mathcal F}
  \left|
  \E_t \left[ f(X) \right]-\E_w \left[ f(X) \right]
  \right|.
\]
The balancing-function class determines the meaning of balance.  
In applications we typically use a prespecified finite balancing-function class
\(\mathcal F=\{f_1,\ldots,f_m\}\).  The balancing functions should represent covariate
discrepancies that are relevant to the intended downstream use.  A small balancing-function class gives an interpretable and easier-to-confirm claim, whereas a richer
balancing-function class gives a stronger diagnostic but requires more target inputs and a
more conservative simultaneous confidence sequence.
\cref{tab:balancing-function-examples} lists balancing functions that recur in
practice, from coordinate means and scales to region indicators and kernel
landmarks.

\begin{table}[H]
\centering
\caption{Examples of balancing functions for finite covariate balance confirmation.  
The balancing-function class should be chosen before monitoring and should encode covariate discrepancies 
that matter for the downstream use.}
\label{tab:balancing-function-examples}
\renewcommand{\arraystretch}{1.15}
\resizebox{\linewidth}{!}{%
\begin{tabular}{@{}lll@{}}
\toprule
\textbf{Balance goal} & \textbf{Typical balancing functions \(f_j(x)\)} & \textbf{What is checked} \\
\midrule
Mean balance & \(x_j\) or clipped standardized coordinates & Important covariate means \\
Scale balance & \(x_j^2\) or tail-transformed coordinates & Marginal scale and tail mismatch \\
Interaction balance & \(x_jx_k\) or selected products & Low-order dependence structure \\
Model-relevant balance & Representation coordinates \(\phi_j(x)\) & Features used by a predictor or surrogate \\
Region balance & \(\one\{x\in A_j\}\) & Subgroups, bins, support regions, or tails \\
Kernel or local balance & \(k(x,z_j)\) for landmarks \(z_j\) & Local distributional mismatch or MMD approximations \\
\bottomrule
\end{tabular}%
}
\end{table}

For unbounded or heavy-tailed balancing functions, the functions should be standardized and
clipped to a bounded range before using bounded confidence sequences, or handled
with confidence sequences matched to the assumed tail behavior.  All balancing
functions and tolerances should be fixed before the monitoring stream is used;
choosing them after inspecting the same target inputs would require additional
post-selection control.

\subsection{Finite balancing-function class}
\label{subsec:finite-balancing-functions}

We develop covariate balance confirmation for a finite balancing-function class
\[
  \mathcal F=\{f_1,\ldots,f_m\},
\]
with finite target and corrected expectations for every \(f_j\).  Boundedness is
needed only for the explicit Hoeffding construction below; other confidence
sequences may be used under corresponding tail assumptions.
Let
\[
  \mu_{t,j}=\E_t  \left[f_j(X)\right],
  \qquad
  \mu_{w,j}=\E_w \left[f_j(X)\right].
\]
Then
\[
  \Delta_{\mathcal F}(w)
  =
  \max_{j\le m}|\mu_{t,j}-\mu_{w,j}|.
\]
Thus a vector of per-function tolerances
\(\varepsilon_1,\ldots,\varepsilon_m\) defines the finite-dimensional balance
claim
\begin{equation}
  |\mu_{t,j}-\mu_{w,j}|\le \varepsilon_j,
  \qquad j=1,\ldots,m.
  \label{eq:finite-balance-claim}
\end{equation}
Taking \(\varepsilon_j\equiv\varepsilon\) recovers
\(\Delta_{\mathcal F}(w)\le\varepsilon\).

As target inputs arrive, an ordinary confidence interval computed at one
prespecified sample size is not sufficient for continuous monitoring.  Recomputing
fixed-time intervals after every new target input and stopping when one first fits
inside a tolerance band can invalidate the nominal coverage guarantee.  A
\emph{confidence sequence} is designed for this setting: it is a sequence of
intervals that covers the same population quantity simultaneously over all
monitoring indices.  For each balancing function, let
\[
  C_{n,j}=[L_{n,j},U_{n,j}]
\]
denote the interval after the first \(n\) target inputs.  Allocating the error
probability jointly over monitoring indices and balancing functions gives
\begin{equation}
  \Prob\left\{
  \forall n\ge 1,\ \forall j\le m:
  \mu_{t,j}\in C_{n,j}
  \right\}\ge 1-\delta.
  \label{eq:simultaneous-cs}
\end{equation}
This event is stronger than coverage at any single fixed sample size.  On this
event, the analyst may inspect the intervals after every target input and stop at
the first data-dependent monitoring index when all required containment
conditions hold, without inflating the false-confirmation probability.  This
optional-stopping validity is the reason confidence sequences, rather than
ordinary confidence intervals, are used in this section.  For iid bounded balancing-function values,
\eqref{eq:simultaneous-cs} can be obtained from a Hoeffding construction, while
sub-Gaussian or variance-adaptive constructions apply under other tail
conditions.  Beyond iid sampling, the theorem remains valid whenever
\eqref{eq:simultaneous-cs} holds for the fixed target means under the actual
sequential distribution.

The empirical discrepancy
\[
  |\widehat\mu_{n,j}-\widehat\mu_{w,j}|
\]
is still useful descriptively, but it is not by itself an anytime-valid rule for
the population discrepancy
\[
  |\mu_{t,j}-\mu_{w,j}|.
\]
Here \(\widehat\mu_{n,j}=n^{-1}\sum_{i=1}^n f_j(X_i)\) is the running target
average, and \(\widehat\mu_{w,j}\) is the corresponding weighted source average
when the source expectation is estimated empirically.  The observed gap can look
small because of target-sampling noise, particularly early in monitoring, or
because repeated inspection selects the first favorable fluctuation.  The
confidence sequence converts this running diagnostic into a population-level
sequential confirmation rule: stopping is allowed only when every target mean that
remains plausible under \(C_{n,j}\) lies inside its prescribed tolerance band.

When the weighted source moments \(\mu_{w,j}\) are known, the confirmation rule
is simply containment in the fixed tolerance bands:
\begin{equation}
  C_{n,j}
  \subseteq
  [\mu_{w,j}-\varepsilon_j,\mu_{w,j}+\varepsilon_j],
  \qquad j=1,\ldots,m.
  \label{eq:balance-stopping}
\end{equation}
This containment rule is conservative in the right direction: it waits until all
plausible target means, according to the confidence sequences, are inside the
prespecified tolerance bands.

\bigskip
\begin{theorem}[Anytime-valid covariate balance confirmation]
\label{thm:balance-confirmation}
Assume \eqref{eq:simultaneous-cs}.  Define the stopping time
\[
  \tau_{\mathrm{bal}}
  =
  \inf\left\{
  n\ge 1:
  C_{n,j}\subseteq
  [\mu_{w,j}-\varepsilon_j,\mu_{w,j}+\varepsilon_j]
  \text{ for all }j\le m
  \right\}.
\]
If the correction is out of tolerance for at least one balancing function, meaning that
\[
  |\mu_{t,j}-\mu_{w,j}|>\varepsilon_j
  \quad\text{for some }j,
\]
then
\[
  \Prob\left\{\tau_{\mathrm{bal}}<\infty\right\}\le \delta.
\]
\end{theorem}

\noindent\emph{Proof sketch.}
On the simultaneous confidence-sequence event, every target balancing-function mean remains
inside its confidence interval at all monitoring indices.  If any balancing function is outside the
allowed tolerance band, no confidence interval containing the true target mean
can be contained in that band.  The full proof is in
Appendix~\ref{app:proof-balance-confirmation}.

\bigskip
Theorem~\ref{thm:balance-confirmation} provides a false-confirmation guarantee for a sequential balance-confirmation procedure.  
The analyst may continue monitoring target inputs and stop at the first time when all target confidence sequences 
are contained in their prescribed tolerance bands.  
If the true discrepancy exceeds its tolerance for at least one balancing function, then the probability 
that the procedure ever stops and incorrectly confirms the correction is at most \(\delta\).  
The key point is that a small empirical gap alone does not constitute valid balance confirmation.  
Confirmation requires every value of each target balancing-function mean that remains plausible 
under the confidence sequence to lie inside the corresponding tolerance band.  
In this sense, the theorem converts a running empirical diagnostic into a valid confirmation decision and, upon stopping, a population-level certificate.  
It is not, however, a power guarantee: even when the correction is truly within tolerance, 
confirmation may take a long time if the tolerances are tight, the balancing-function variance is large, or the target stream is short.

\begin{remark}[What makes confirmation anytime-valid]
\label{rem:balance-anytime-valid}
The anytime-validity of covariate balance confirmation does not require the
balance statistic itself to be an e-process.  It follows from the simultaneous
coverage event
\[
  \Prob\left\{
    \forall n\ge1,\ \forall j\le m:
    \mu_{t,j}\in C_{n,j}
  \right\}\ge1-\delta.
\]
On this event, if any true balancing-function mean lies outside its tolerance
band, no confidence sequence containing that mean can ever be contained in the
band.  Hence false confirmation can occur only if the simultaneous confidence
sequence fails, with probability at most \(\delta\).  The certificate is
therefore valid under continuous monitoring and optional stopping, although its
scope remains restricted to the prespecified balancing functions and
tolerances.
\end{remark}

\subsection{Choices of time-uniform confidence sequences}
\label{subsec:hoeffding-cs}

Theorem~\ref{thm:balance-confirmation} does not depend on a particular
confidence-sequence construction.  Any sequence satisfying the simultaneous
coverage requirement in \eqref{eq:simultaneous-cs} may be used.  The choice
changes the interval width and hence the confirmation time, but not the logic of
the balance-confirmation procedure.  We summarize several representative options.

\medskip
\noindent\textbf{Hoeffding confidence sequence.}
For a self-contained implementation, suppose that each balancing-function value satisfies
\(f_j(X)\in[0,1]\) almost surely, and let
\[
  \widehat\mu_{n,j}=\frac{1}{n}\sum_{i=1}^n f_j(X_i).
\]
A union-over-time Hoeffding sequence is
\[
  C_{n,j}^{\mathrm{H}}
  =
  \left[
  \widehat\mu_{n,j}-r_{n,j}^{\mathrm{H}},
  \widehat\mu_{n,j}+r_{n,j}^{\mathrm{H}}
  \right]\cap[0,1],
\]
where
\begin{equation}
  r_{n,j}^{\mathrm{H}}
  =
  \sqrt{
  \frac{\log\{2m\pi^2n^2/(6\delta)\}}{2n}
  }.
  \label{eq:hoeffding-radius}
\end{equation}
This construction requires only boundedness.  It is transparent and easy to
reproduce, but it uses the worst-case range and can therefore be conservative
when the balancing-function variance is small.

\medskip
\noindent\textbf{Other time-uniform alternatives.}
The Hoeffding construction is not essential to our balance-confirmation theorem.
Any confidence sequence satisfying the required time-uniform coverage guarantee
may be substituted.  Empirical-Bernstein confidence sequences adapt their width
to an accumulated empirical or predictable variance and can therefore be
narrower than range-based Hoeffding sequences when the balancing-function values
have low variance; see \citet{HowardSR2021aos}.  Betting confidence sequences
construct nonnegative wealth processes for candidate means and invert those
processes to obtain time-uniform confidence sets; see
\citet{Waudby-SmithI2024jrsssb}.  These alternatives can improve statistical
efficiency, but they require additional choices of a uniform boundary, betting
strategy, or mixture.  In the Gaussian experiments of \cref{sec:experiments},
the coordinate balancing functions are unbounded and we therefore use the
explicit sub-Gaussian confidence sequence stated in
\cref{subsec:exp-balance}, rather than the bounded Hoeffding sequence.
\Cref{tab:cs-choices} summarizes these representative choices and their roles
in the present paper.

\begin{table}[t]
\centering
\caption{Representative confidence-sequence choices for target
balancing-function means.  The balance-confirmation theorem accepts any choice
that provides the required simultaneous time-uniform coverage.}
\label{tab:cs-choices}
\renewcommand{\arraystretch}{1.12}
\begin{tabular}{@{}>{\raggedright\arraybackslash}p{0.21\linewidth}>{\raggedright\arraybackslash}p{0.37\linewidth}>{\raggedright\arraybackslash}p{0.34\linewidth}@{}}
\toprule
\textbf{Construction} & \textbf{Main feature} & \textbf{Role in this paper} \\
\midrule
Hoeffding & Explicit and requires boundedness & Bounded-function reference construction \\
Sub-Gaussian & Explicit under a known variance proxy & Used in the Gaussian experiments \\
Empirical Bernstein & Adapts to accumulated variance & Optional efficiency improvement \\
Betting or mixture & Obtained by inverting wealth processes & Optional flexible alternative \\
\bottomrule
\end{tabular}
\end{table}

The validity result is agnostic to this choice.  It requires only that each
\(C_{n,j}\) be a time-uniform confidence sequence and that the simultaneous
error across the balancing functions be controlled as in
\eqref{eq:simultaneous-cs}.  Changing the confidence-sequence construction changes
the target-side uncertainty widths and hence the stopping times, but not the logic
or false-confirmation guarantee of the procedure.

The selected confidence sequence supplies the target-side uncertainty used by
both the known-source and finite-source balance procedures.  We next account for
the additional uncertainty created when normalization constants and weighted
source moments must be estimated from a finite source sample.

\subsection{Finite-source samples and normalization}
\label{sec:finite-source}

So far, the balance rule treats \(\E_s \left[w(X)\right]\) and \(\E_w \left[f(X)\right]\) as known.  In
practice, the source distribution is represented by a finite sample.  This
creates two distinct finite-source issues.  The first concerns the global
monitor of \cref{sec:eprocess}: its increment must have source expectation at
most one, so an unnormalized correction needs a conservative normalizer.  The
second concerns covariate balance confirmation: the weighted source moment that
centers each tolerance band is itself estimated.  These are separate problems
and require different remedies.

\subsubsection{Conservative normalizer for the e-process}
\label{subsec:conservative-normalizer}

This subsection concerns only the global e-process, not the balance-confirmation
procedure.  If the proposed correction is initially available in unnormalized
form \(\widetilde w\), the ideal normalizer is
\[
  Z=\E_s \left[\widetilde w(X)\right].
\]
If \(Z\) is unknown and replaced by a noisy empirical average, the global
monitor may no longer be exactly valid under the source-reference calibration
null.  A conservative repair is to use an upper confidence bound \(U_Z\) such
that
\[
  \Prob\left\{Z\le U_Z\right\}\ge 1-\eta.
\]
Then
\[
  e_n=\frac{\widetilde w(X_n)}{U_Z}
\]
has source expectation at most one on the event \(Z\le U_Z\).  Therefore the
resulting process is calibrated up to the source-normalization failure
probability.  This normalization issue is separate from the source-moment
uncertainty used below for covariate balance confirmation.

A simple implementation is available when \(0\le \widetilde w(x)\le B\) is known
for some finite \(B>0\), with \(\eta\in(0,1)\).  Conditional on all data used to construct \(\widetilde w\), let
\(X_1^{s,0},\ldots,X_{n_0}^{s,0}\stackrel{\mathrm{iid}}{\sim}P_s^X\) be an
independent source split reserved for normalization and define
\[
  \overline W_{n_0}
  =
  \frac{1}{n_0}\sum_{r=1}^{n_0}\widetilde w(X_r^{s,0}).
\]
A one-sided Hoeffding upper confidence bound is
\begin{equation}
  U_Z
  =
  \min\left\{
    B,
    \overline W_{n_0}
    +B\sqrt{\frac{\log(1/\eta)}{2n_0}}
  \right\},
  \label{eq:normalizer-hoeffding}
\end{equation}
which satisfies \(\Prob\left\{Z\le U_Z\right\}\ge1-\eta\), conditionally on the
construction data and hence also unconditionally.  The correction and \(U_Z\) are
then frozen before target monitoring.  For unbounded weights, a finite-sample
upper bound requires an explicit tail or moment assumption, or an intentionally
clipped correction; clipping changes the correction being monitored.

\bigskip
\begin{proposition}[E-process with conservative normalizer]
\label{prop:conservative-normalizer}
Suppose the nonnegative function \(\widetilde w\) is fixed before monitoring
and is \(\cG_0\)-measurable, with
\(Z=\E_s[\widetilde w(X)]\) satisfying \(0<Z<\infty\) almost surely.
Let \(\eta\in(0,1)\), and suppose \(U_Z\) is computed before monitoring, is
\(\cG_0\)-measurable, satisfies \(0<U_Z<\infty\), and obeys
\(\Prob\left\{Z\le U_Z\right\}\ge 1-\eta\).  Under the source-reference calibration null,
conditional on \(\cG_0\), for every realization satisfying \(Z\le U_Z\),
\[
  M_n=\prod_{i=1}^n \frac{\widetilde w(X_i)}{U_Z}
\]
is a nonnegative supermartingale with respect to \((\cG_n)_{n\ge0}\) under the corresponding conditional law.  Hence
\[
  \Prob_{H_0^{\mathrm{src}}}
  \left\{
  \sup_{n\ge0} M_n\ge \frac{1}{\alpha}
  \right\}
  \le \alpha+\eta.
\]
\end{proposition}

\noindent\emph{Proof sketch.}
On the event \(Z\le U_Z\), the normalized increment
\(\widetilde w(X)/U_Z\) has source expectation at most one, so the product is a
supermartingale.  Ville's inequality and the failure probability of the source
bound give the result.  The full proof is in
Appendix~\ref{app:proof-conservative-normalizer}.  The Gaussian experiments use
analytically normalized exponential weights, so this conservative normalizer is
not needed there.

\begin{remark}[Confirmation-only normalizer]
The conservative normalizer is designed for the global relative monitoring
boundary \(M_n\ge 1/\alpha\).  It should be interpreted only as a one-sided
confirmation device: replacing the exact normalizer by an upper confidence bound
makes the increments conservative under the source-reference calibration null, but it does not
turn failure to cross the boundary into evidence against the proposed correction.
\end{remark}

\subsubsection{Weighted source moment uncertainty}
\label{subsec:source-moment-uncertainty}

For covariate balance confirmation, the weighted source moment
\[
  \mu_{w,j}
  =
  \frac{\E_s  \left[\widetilde w(X)f_j(X)\right]}{\E_s \left[\widetilde w(X)\right]}
\]
is also often estimated.  Suppose a source split independent of the data used to construct the correction,
or another conditionally valid construction, gives simultaneous source confidence
intervals
\[
  \mu_{w,j}\in[\ell_{w,j},u_{w,j}],
  \qquad j=1,\ldots,m,
\]
with probability at least \(1-\eta\).  The target confidence sequences still
cover \(\mu_{t,j}\) simultaneously over time and balancing functions with probability at
least \(1-\delta\).  On the joint good event, finite source uncertainty leads to
two distinct questions.

First, can we confirm the original claim
\[
  |\mu_{t,j}-\mu_{w,j}|\le \varepsilon_j
\]
even though the exact value of \(\mu_{w,j}\) is unknown?  Second, if such confirmation is not yet possible, is the target behavior at least compatible with
\(\varepsilon_j\)-balance for some weighted source moment allowed by the source
confidence interval?  The first question concerns formal confirmation.  The
second concerns only a source-uncertainty diagnostic.  They lead to different
bands because they use different quantifiers over the unknown source moment.

\noindent\textbf{Confirmation band: close to every plausible source moment.}
To confirm the original \(\varepsilon_j\)-balance claim, the target confidence
sequence must lie within \(\varepsilon_j\) of the true \(\mu_{w,j}\), regardless
of where that true value lies in \([\ell_{w,j},u_{w,j}]\).  The appropriate band
is therefore the intersection of all tolerance bands centered at plausible
weighted source moments:
\begin{align}
  B^{\mathrm{conf}}_j
  &:=
  \bigcap_{\nu\in[\ell_{w,j},u_{w,j}]}
  [\nu-\varepsilon_j,\nu+\varepsilon_j] \\
  &=
  [u_{w,j}-\varepsilon_j,\ \ell_{w,j}+\varepsilon_j].
  \label{eq:inner-band}
\end{align}
The containment rule
\[
  C_{n,j}\subseteq B^{\mathrm{conf}}_j
\]
means that every value of the target mean still plausible at monitoring index \(n\) is within
\(\varepsilon_j\) of every weighted source mean allowed by the source interval.
In particular, it is within \(\varepsilon_j\) of the true \(\mu_{w,j}\).
Consequently, requiring this containment for all balancing functions confirms the original
\(\varepsilon\)-balance claim with false-confirmation probability at most
\(\delta+\eta\).
The confirmation band is the band used by the proposed algorithm.  It contracts
as source uncertainty grows and can be empty when
\[
  u_{w,j}-\ell_{w,j}>2\varepsilon_j.
\]
An empty confirmation band is not evidence that the correction is poor.  It says
that the source sample is too uncertain to confirm the requested tolerance for
that balancing function.  More source information, a wider tolerance, or a less variable
correction is then required before exact confirmation is possible.

\bigskip
\begin{corollary}[Finite-source covariate balance confirmation]
\label{cor:finite-source-balance}
Suppose the source intervals satisfy
\[
  \Prob\left\{
    \mu_{w,j}\in[\ell_{w,j},u_{w,j}]
    \text{ for all }j\le m
  \right\}\ge1-\eta,
\]
and the target confidence sequences satisfy \eqref{eq:simultaneous-cs}.  Define
\[
  \tau_{\mathrm{bal}}^{\mathrm{fs}}
  =
  \inf\left\{
    n\ge1:
    C_{n,j}\subseteq B_j^{\mathrm{conf}}
    \text{ for all }j\le m
  \right\}.
\]
If \(|\mu_{t,j}-\mu_{w,j}|>\varepsilon_j\) for some \(j\), then
\[
  \Prob\left\{\tau_{\mathrm{bal}}^{\mathrm{fs}}<\infty\right\}
  \le \delta+\eta.
\]
\end{corollary}

\noindent\emph{Proof sketch.}
On the joint event that all source intervals and all target confidence sequences
cover, the confirmation band is contained in the tolerance band around the true
\(\mu_{w,j}\).  An out-of-tolerance target mean therefore prevents containment at
every monitoring index.  The complement of the joint event has probability at
most \(\delta+\eta\).  The full proof is in
Appendix~\ref{app:proof-finite-source-balance}.

\noindent\textbf{Compatibility band: close to at least one plausible source
moment.}
A natural practical alternative is the union of the same tolerance bands:
\begin{align}
  B^{\mathrm{comp}}_j
  &:=
  \bigcup_{\nu\in[\ell_{w,j},u_{w,j}]}
  [\nu-\varepsilon_j,\nu+\varepsilon_j] \\
  &=
  [\ell_{w,j}-\varepsilon_j,\ u_{w,j}+\varepsilon_j].
  \label{eq:expanded-band}
\end{align}
Containment
\[
  C_{n,j}\subseteq B^{\mathrm{comp}}_j
\]
means that every currently plausible target mean lies within
\(\varepsilon_j\) of at least one weighted source moment compatible with the
finite source interval.  This can be useful operationally, but it does not establish closeness to the true \(\mu_{w,j}\), which may lie near the opposite end
of the interval.  On the joint good event, it supports at most the relaxed bound
\[
  |\mu_{t,j}-\mu_{w,j}|
  \le
  \varepsilon_j+(u_{w,j}-\ell_{w,j}),
\]
not the original \(\varepsilon_j\)-balance statement.
The confirmation band is an intersection because the confirmed claim must hold for
every source moment still plausible under the source confidence interval.  The
compatibility band is a union because it asks only whether some plausible source
moment could explain the target behavior.  \Cref{tab:finite-source-bands}
contrasts the two bands and their different inferential meanings.

\begin{table}[t]
\centering
\caption{Two finite-source bands for balancing function \(f_j\).  The confirmation band is
used by the proposed procedure.  The compatibility band is optional diagnostic
information and does not confirm the original \(\varepsilon_j\)-balance claim.}
\label{tab:finite-source-bands}
\renewcommand{\arraystretch}{1.15}
\resizebox{\linewidth}{!}{%
\begin{tabular}{@{}llll@{}}
\toprule
\textbf{Band} & \textbf{Set operation} & \textbf{Quantifier over source mean} &
\textbf{Interpretation} \\
\midrule
\(B_j^{\mathrm{conf}}=[u_{w,j}-\varepsilon_j,\ell_{w,j}+\varepsilon_j]\) &
Intersection & For every plausible \(\mu_{w,j}\) &
Formal \(\varepsilon_j\)-balance confirmation \\
\(B_j^{\mathrm{comp}}=[\ell_{w,j}-\varepsilon_j,u_{w,j}+\varepsilon_j]\) &
Union & For at least one plausible \(\mu_{w,j}\) &
Source-uncertainty compatibility diagnostic \\
\bottomrule
\end{tabular}%
}
\end{table}

A one-dimensional example illustrates the distinction.  Suppose
\([\ell_{w,j},u_{w,j}]=[0,1]\) and \(\varepsilon_j=0.2\).  The compatibility
band is \([-0.2,1.2]\).  A target mean near \(1.2\) is compatible with a source
moment near \(1\), but if the true weighted source mean is near \(0\), the true
discrepancy is about \(1.2\), far larger than \(0.2\).  The confirmation band is
\([0.8,0.2]\), which is empty.  The correct conclusion is therefore not that the
correction has failed, but that the source data are insufficient to confirm
\(0.2\)-balance.

The two bands thus play complementary but unequal roles.  The confirmation band
answers the paper's inferential question and is used for formal confirmation.  The
compatibility band is included only to diagnose source-limited evidence: it can
show that the target stream remains consistent with the correction even when
finite source uncertainty prevents formal confirmation.  Reporting the latter
must therefore use language such as ``compatible with \(\varepsilon\)-balance,''
not ``\(\varepsilon\)-balance confirmed.''

\subsubsection{Effective sample size}
\label{subsec:ess}

Even if the correction is valid, highly variable weights make weighted source
moments unstable.  A useful diagnostic is the effective sample size
\[
  \ess(w)
  =
  \frac{\left(\sum_{i=1}^{n_s}w_i\right)^2}
       {\sum_{i=1}^{n_s}w_i^2}.
\]
Small \(\ess(w)\) does not invalidate the target-input e-process, but it weakens
downstream weighted calibration and widens source moment intervals.  In practice,
confirmation should report both relative evidence and the effective
sample size of the correction.

\subsection{Unified anytime-valid balance procedure}
\label{subsec:balance-algorithm}

All quantities used by the procedure have now been defined.  With known weighted
source moments, the confirmation and compatibility bands coincide with the fixed
tolerance band around \(\mu_{w,j}\).  With finite-source moments, the procedure
uses the contracted confirmation band from
\cref{subsec:source-moment-uncertainty} for formal confirmation and may report
the expanded compatibility band only as a diagnostic.  \Cref{alg:balance}
collects the two cases in one implementation.

\begin{algorithm}[ht!]
\caption{Anytime-valid covariate balance confirmation using balancing functions}
\label{alg:balance}
\begin{algorithmic}[1]
\Require Prespecified correction \(w\), balancing functions \(f_1,\ldots,f_m\),
tolerances \(\varepsilon_1,\ldots,\varepsilon_m\), target level \(\delta\),
source level \(\eta\)
\If{weighted source moments \(\mu_{w,j}=\E_w \left[ f_j(X) \right]\) are known}
  \State Set confirmation bands \(B^{\mathrm{conf}}_j=[\mu_{w,j}-\varepsilon_j,\mu_{w,j}+\varepsilon_j]\)
  \State Set compatibility bands \(B^{\mathrm{comp}}_j=B^{\mathrm{conf}}_j\)
  \State False-confirmation level is \(\delta\)
\Else
  \State Build conditionally valid simultaneous source confidence intervals
  \([\ell_{w,j},u_{w,j}]\) for \(\mu_{w,j}\) at level \(1-\eta\)
  \State Set confirmation bands
  \(B^{\mathrm{conf}}_j=[u_{w,j}-\varepsilon_j,\ \ell_{w,j}+\varepsilon_j]\)
  \State Set compatibility bands
  \(B^{\mathrm{comp}}_j=[\ell_{w,j}-\varepsilon_j,\ u_{w,j}+\varepsilon_j]\)
  \State If any \(B^{\mathrm{conf}}_j\) is empty, \(\varepsilon\)-balance cannot be confirmed from the current source information
  \State False-confirmation level for \(\varepsilon\)-balance is \(\delta+\eta\)
\EndIf
\For{\(n=1,2,\ldots\)}
  \State Observe target input \(X_n\)
  \For{\(j=1,\ldots,m\)}
    \State Update \(\widehat\mu_{n,j}=n^{-1}\sum_{i=1}^n f_j(X_i)\)
  \EndFor
  \State Update \(C_{n,1},\ldots,C_{n,m}\) using a jointly valid construction satisfying \eqref{eq:simultaneous-cs}
  \If{\(C_{n,j}\subseteq B^{\mathrm{conf}}_j\) for all \(j\)}
    \State \Return confirm \(\varepsilon\)-balance of \(w\)
  \ElsIf{\(C_{n,j}\subseteq B^{\mathrm{comp}}_j\) for all \(j\)}
    \State Optionally report compatibility only
    \Statex \hspace{\algorithmicindent} Do not report \(\varepsilon\)-balance confirmation
  \EndIf
\EndFor
\end{algorithmic}
\end{algorithm}

\subsection{Reporting global evidence and balance confirmation together}
\label{subsec:combined}

The global monitor and the balance confirmation procedure are procedurally
separate: neither consumes the other's output, although both are computed from
the same target stream and are therefore generally statistically dependent.
Each has its own validity statement.  They are run in parallel and reported
together because they provide nonredundant but individually incomplete
information.  Their blind spots can overlap, so neither should be treated as a
substitute for a sufficiently rich balancing-function class.  A natural
reporting order, from cheapest and coarsest to strongest, is
\[
\begin{gathered}
  \underbrace{\text{source-calibrated global evidence}}_{\text{direction, no restriction to }\mathcal F}
  \;;\;
  \underbrace{\varepsilon\text{-balance confirmation}}_{\text{tolerance, local to }\mathcal F}
  \\[4pt]
  \longrightarrow\;
  \text{downstream weighted conformal calibration},
\end{gathered}
\]
where the semicolon marks two parallel checks and only the arrow marks a genuine
dependence: in the proposed workflow, downstream deployment is gated on the
tolerance-level confirmation decision and resulting certificate.  Source-calibrated global evidence can accumulate quickly
when the target stream is clearly not source-like, while sustained negative
log-growth is diagnostic of a globally harmful direction.  Covariate balance
confirmation usually takes longer because it must shrink confidence sequences
until the remaining uncertainty fits inside the tolerance bands.  This asymmetry
reflects the different inferential targets: directional relative evidence can be
available before the procedure can issue a tolerance-level adequacy certificate.

The crossing event \(M_n\ge 1/\alpha\) provides source-calibrated evidence for
the proposed weighted input distribution relative to the unweighted source,
without restricting attention to the selected balancing-function class.  The balance stopping rule is quantitatively stronger but
narrower: within the chosen balancing-function class and tolerance, it supports
the claim that weighted source covariates match target covariates closely
enough.  This second statement is the one that matters when the correction is
used as a downstream reliability claim.  The first is a broader relative
monitor that can warn about some globally harmful correction directions while a
limited balancing-function class looks satisfied.  It need not detect an
incomplete correction that nevertheless improves on the source in KL risk.

\section{Acceptable-Shift Monitoring and Downstream Deployment}
\label{sec:tolerance}

The procedures developed above support two distinct operational decisions.  The
first is a regime-monitoring decision: whether the incoming target stream has
moved beyond a prespecified family of covariate shifts that the application is
prepared to handle.  The second is a deployment decision: once a particular
correction has received the required covariate-balance confirmation, whether and
how it should enter a downstream importance-weighted procedure.  We illustrate
these two uses with an acceptable exponential-tilt region and weighted conformal
prediction.  They are complementary.  The acceptable-region test can warn that
no correction in a manageable family explains the target stream, whereas
balance confirmation evaluates a particular correction on the functions and
tolerances required for downstream deployment.

\subsection{Testing whether covariate shift exceeds an acceptable region}
\label{subsec:exp-family}

In many deployments, the relevant question is not whether the target input
distribution differs from the source at all.  Moderate covariate shift may be
expected and manageable.  The operational concern is whether the target stream
has moved beyond a prespecified regime for which available corrections, models,
or operating policies are regarded as adequate.  We formulate this question
using an exponential-tilt family and construct an anytime-valid test of the
composite null that some correction in the acceptable region explains the
incoming stream.

Consider the exponential-tilt family
\begin{equation}
  w_\theta(x)
  =
  \exp\left\{\theta^\top\phi(x)-\psi_s(\theta)\right\},
  \qquad
  \psi_s(\theta)
  =
  \log \E_s\!\left[\exp\left\{\theta^\top\phi(X)\right\}\right],
  \label{eq:exp-tilt}
\end{equation}
where \(\phi(x)\) is a prespecified feature vector and \(\theta\) belongs to the
natural parameter domain of \(\psi_s\).  The weight \(w_\theta\) induces the
corrected input distribution
\[
  dP_\theta^X(x)=w_\theta(x)\,dP_s^X(x).
\]
For illustration, let
\[
  \Theta_0=\{\theta:\|\theta\|\le\kappa\}
\]
be the set of corrections regarded as operationally acceptable.  The feature
map, its scaling, the norm, and \(\kappa\) must be fixed before monitoring and
should be justified by domain requirements or historical operating ranges.  A
raw parameter-radius interpretation is not invariant to rescaling of
\(\phi\).  When more interpretable, the acceptable set can instead be specified
through a divergence limit, a moment-deviation constraint, or another compact
or convex criterion tied to downstream requirements.

The composite sequential null hypothesis is
\begin{equation}
  H_0^{\mathrm{tol}}:
  \quad
  \text{there exists a fixed }\theta\in\Theta_0\text{ such that }
  X_i\mid\cG_{i-1}\sim P_\theta^X
  \text{ for every }i\ge1.
  \label{eq:exp-tilt-null}
\end{equation}
Thus, under the null, at least one correction in the acceptable region explains
the target stream.  Rejecting \cref{eq:exp-tilt-null} gives evidence that the
entire acceptable exponential-tilt region is inadequate for the incoming
stream.  It does not identify the true correction parameter.

To test this composite null, choose a fixed vector \(\lambda\) before observing
the target stream.  The vector \(\lambda\) is a \emph{test direction}: it
specifies the feature-space direction in which departures beyond the acceptable
region will be probed.  It is not itself a correction parameter being tested.
A scientifically motivated alternative or an externally proposed correction
outside \(\Theta_0\) may suggest a useful direction, but the test does not assert
that this alternative is correct.  Choosing \(\lambda\) after inspecting the
monitoring data generally invalidates the fixed-direction guarantee unless
sample splitting or a valid mixture construction is used.

Choose \(\lambda\) so that \(\theta+\lambda\) lies in the natural parameter
domain of \(\psi_s\) for every \(\theta\in\Theta_0\).  Then the
exponential-family identity gives
\[
  \log \E_{P_\theta^X}\left[
  \exp\{\lambda^\top\phi(X)\}
  \right]
  =
  \psi_s(\theta+\lambda)-\psi_s(\theta).
\]
Define the worst-case log moment-generating function over the acceptable region
by
\begin{equation}
  \psi_{\Theta_0}(\lambda)
  \defeq
  \sup_{\theta\in\Theta_0}
  \left\{
    \psi_s(\theta+\lambda)-\psi_s(\theta)
  \right\},
  \qquad
  \psi_{\Theta_0}(\lambda)<\infty.
  \label{eq:composite-log-mgf}
\end{equation}
The one-step e-value in direction \(\lambda\) is
\begin{equation}
  e_i(\lambda)
  =
  \exp\left\{
    \lambda^\top\phi(X_i)-\psi_{\Theta_0}(\lambda)
  \right\}.
  \label{eq:exp-tilt-evalue}
\end{equation}
For every null value \(\theta\in\Theta_0\),
\[
  \E_{P_\theta^X}\!\left[e_i(\lambda)\mid\cG_{i-1}\right]
  \le1.
\]
Consequently,
\begin{equation}
  M_n(\lambda)
  =
  \prod_{i=1}^n e_i(\lambda)
  \label{eq:exp-tilt-eprocess}
\end{equation}
is a nonnegative supermartingale under every distribution in the composite null.
In particular,
\begin{equation}
  \sup_{\theta\in\Theta_0}
  \Prob_{P_\theta^X}\!\left\{
    \sup_{n\ge1}M_n(\lambda)\ge\frac{1}{\alpha}
  \right\}
  \le\alpha.
  \label{eq:exp-tilt-validity}
\end{equation}
A boundary crossing therefore gives anytime-valid evidence against the entire
acceptable-region null, as witnessed in the prespecified direction
\(\lambda\).  It should not be interpreted as evidence that \(\lambda\) is a
true parameter or that the target distribution has been identified.
Conversely, failure to cross is inconclusive and does not establish that the
target remains inside the acceptable region.

The choice of \(\lambda\) determines power.  If the monitoring distribution is
\(P_{\theta_1}^X\), where \(\theta_1\) lies in the interior of the natural
parameter domain and \(\psi_s\) is differentiable at \(\theta_1\), then
\begin{equation}
  \E_{P_{\theta_1}^X}\!\left[\log e_i(\lambda)\right]
  =
  \lambda^\top\nabla\psi_s(\theta_1)
  -
  \psi_{\Theta_0}(\lambda).
  \label{eq:composite-tilt-drift}
\end{equation}
Positive asymptotic growth requires
\[
  \lambda^\top\nabla\psi_s(\theta_1)
  >
  \psi_{\Theta_0}(\lambda).
\]
Thus, being outside \(\Theta_0\) does not by itself guarantee rapid detection by
a fixed direction.  The alternative must be sufficiently separated from the
worst-case acceptable member in the chosen feature direction.  When no single
direction is clearly preferred, a probability mixture
\[
  \overline M_n
  =
  \int M_n(\lambda)\,\Pi(d\lambda),
\]
with \(\Pi\) fixed before monitoring and supported on directions for which
\(\psi_{\Theta_0}(\lambda)<\infty\), remains an e-process and spreads power over
multiple prespecified directions.

The punch line is a set-level monitoring statement: a crossing provides
time-uniform evidence that the entire prespecified family of acceptable
corrections is insufficient for the target stream, as detected by the chosen
direction or mixture.  This is distinct from covariate balance confirmation,
which asks whether a particular correction satisfies absolute tolerances for
selected balancing functions and can therefore be deployed in a specified
downstream task.

\subsection{Weighted conformal prediction after balance confirmation}
\label{sec:weighted-conformal}

We next explain how confirmed covariate corrections enter a downstream task that
is standard for conformal prediction readers: weighted split conformal
calibration under covariate shift.  The point of this subsection is not to
introduce a new conformal procedure.  Weighted conformal calibration already
specifies how an input density ratio should enter the calibration quantile.  Our
contribution is the preceding confirmation step: before using a proposed
correction \(w\) in that quantile, we ask whether it is credible for the
incoming target stream and for the downstream balancing functions.

Let \(s(x,y)\) be any fitted nonconformity score constructed from the training
split \(\Dtr\).  If a predictive distribution \(p(y\mid x,\Dtr)\) is available,
one possible score is
\[
  s(x,y)=-\log p(y\mid x,\Dtr).
\]
No special notation is needed for the predictive distribution; it may be
Bayesian, likelihood-based, ensemble-based, or any other fitted predictive
model.  Throughout this subsection, \(\Dtr\) is independent of the calibration
split and the target test point, and the conformal validity argument is
conditional on \(\Dtr\), so that \(s\) is treated as fixed.  Under pure covariate
shift,
\[
  P_t(Y\mid X)=P_s(Y\mid X),
\]
the conditional prediction problem is unchanged.  The input distribution changes.
Therefore the correction enters through the calibration measure, not through a
predictive tilt in \(y\).  In short,
\[
  \text{confirmed covariate correction}
  \quad\Longrightarrow\quad
  \text{weighted conformal calibration for target inputs}.
\]

\subsubsection{Weighted split conformal calibration}
\label{subsec:weighted-split}

Given calibration data
\[
  \Dcal=\{(X_i^s,Y_i^s)\}_{i=1}^{n_{\cali}},
\]
compute scores
\[
  S_i=s(X_i^s,Y_i^s).
\]
With exact covariate weights \(w^\star=dP_t^X/dP_s^X\), the target score
distribution is represented by the weighted calibration distribution.  Following
\citet{TibshiraniR2019neurips}, exact target coverage requires a \emph{test-point
dependent} normalization that places mass on the test score as well.  For a test
input \(x\), define the normalized weights
\[
  p_i^w(x)
  =
  \frac{w(X_i^s)}{\sum_{j=1}^{n_{\cali}}w(X_j^s)+w(x)},
  \quad i=1,\ldots,n_{\cali},
  \qquad
  p_\infty^w(x)
  =
  \frac{w(x)}{\sum_{j=1}^{n_{\cali}}w(X_j^s)+w(x)},
\]
and the weighted quantile
\[
  \widehat q_w(x)
  =
  \inf\left\{
  q:\;
  \sum_{i=1}^{n_{\cali}}p_i^w(x)\,\one\{S_i\le q\}
  \ge 1-\alpha
  \right\},
\]
where \(\widehat q_w(x)=+\infty\) whenever
\(\sum_i p_i^w(x)<1-\alpha\), i.e.\ the point mass \(p_\infty^w(x)\) at \(+\infty\)
absorbs the remaining probability.  The prediction set is
\[
  C_w(x)=\{y:s(x,y)\le \widehat q_w(x)\}.
\]
Under exact covariate shift with \(w=w^\star\), this construction gives
finite-sample target coverage
\[
  \Prob_t\left\{Y\in C_w(X)\mid\Dtr\right\}
  \ge 1-\alpha.
\]
A simpler \emph{plug-in} threshold drops the test-point mass and, when
\(\sum_{i=1}^{n_{\cali}}w(X_i^s)>0\), uses the fixed weighted quantile
\[
  \widehat q_w^{\,\mathrm{plug}}
  =
  \inf\left\{
  q:
  \frac{\sum_{i=1}^{n_{\cali}}w(X_i^s)\one\{S_i\le q\}}
       {\sum_{i=1}^{n_{\cali}}w(X_i^s)}
  \ge 1-\alpha
  \right\}.
\]
If the calibration-weight sum is zero, set
\(\widehat q_w^{\,\mathrm{plug}}=+\infty\).  This plug-in rule is only
asymptotically exact.  Define the corresponding plug-in set by
\[
  C_w^{\mathrm{plug}}(x)
  =
  \{y:s(x,y)\le \widehat q_w^{\,\mathrm{plug}}\}.
\]
Its finite-sample and approximate-weight error can be related to weighted
empirical-CDF error and score-CDF mismatch, as illustrated by
\cref{prop:coverage-degradation}.

\subsubsection{Validity and efficiency roles}
\label{subsec:validity-efficiency}

Under exact covariate shift and exact weights, weighted conformal calibration is
the coverage mechanism.  Covariate balance confirmation plays a different role:
it is a gatekeeper for the proposed correction.  It asks whether the weights used
by the conformal quantile are credible for the incoming target stream and the
chosen downstream balancing functions.  Confirmation therefore supports the
input weighting step; it does not replace the conformal calibration argument.

The fitted score \(s\) controls the geometry and efficiency of the prediction set.
If the score is obtained from a likelihood-based model, the same weights may also
be used during training.  For example, in a Bayesian model one may form the
weighted posterior
\[
  \pi_w(\theta\mid\Dtr)
  \propto
  \pi_0(\theta)
  \prod_{i\in\Dtr}p_\theta(y_i^s\mid x_i^s)^{w(x_i^s)}.
\]
This can improve efficiency under model misspecification by fitting the
predictive model toward the target covariate region.  But weighted fitting is not
the core validity repair.  The core validity repair under covariate shift remains
weighted calibration under the corrected input distribution.

\noindent\textbf{Covariate shift versus predictive tilting.}
For label shift, the target conditional predictive distribution changes and
predictive tilting is central.  For pure covariate shift, the target conditional
predictive distribution is assumed invariant.  The correction should therefore
enter primarily through the calibration measure via \(w_X(x)\).  Weighted fitting
may improve efficiency under misspecification, but unnecessary predictive tilting
can distort the conditional model.

\subsubsection{Approximate weights and residual score discrepancy}
\label{subsec:approx-weights}

When \(w\) is approximate, target coverage depends on how well the weighted
source score distribution approximates the target score distribution.  Conditional
on the fitted score function, let
\[
  F_t(q)=\Prob_t\left\{s(X,Y)\le q\right\},
\]
and
\[
  F_w(q)=\E_w\left[
  \Prob_s\left\{s(X,Y)\le q\mid X\right\}
  \right].
\]
Under covariate shift, the conditional probability inside the expectation is the
same under source and target.  For exact weights, this yields
\(F_w(q)=F_t(q)\); a direct derivation is given in \cref{app:score-cdf}.
Therefore the discrepancy
\[
  |F_t(q)-F_w(q)|
\]
is an input-distribution discrepancy for the score-relevant balancing function
\[
  f_q(x)=\Prob_s\left\{s(x,Y)\le q\mid X=x\right\}.
\]
This is the downstream link.  If the balancing-function class contains or
approximates the score-relevant functions, then covariate balance confirmation
controls the part of the conformal coverage error caused by residual input
mismatch.  If the class ignores those functions, balance can be confirmed in the
chosen moments while the conformal score distribution remains mismatched.

\bigskip
\begin{proposition}[Coverage degradation from calibration and score-CDF mismatch]
\label{prop:coverage-degradation}
Consider the plug-in threshold \(\widehat q_w^{\,\mathrm{plug}}\) and set
\(C_w^{\mathrm{plug}}\) defined above.  Work conditionally on
\((\Dtr,\Dcal)\).  If
\(\sum_{i=1}^{n_{\cali}}w(X_i^s)=0\), then
\(\widehat q_w^{\,\mathrm{plug}}=+\infty\) by convention and the conditional
coverage is one.  Otherwise, define
\[
  \widehat F_w(q)
  =
  \frac{\sum_{i=1}^{n_{\cali}}w(X_i^s)\one\{S_i\le q\}}
       {\sum_{i=1}^{n_{\cali}}w(X_i^s)}.
\]
If, on the conditioning event,
\[
  \sup_q|\widehat F_w(q)-F_w(q)|\le\gamma,
  \qquad
  \sup_q|F_t(q)-F_w(q)|\le\varepsilon,
\]
then
\[
  \Prob_t\left\{Y\in C_w^{\mathrm{plug}}(X)\mid\Dtr,\Dcal\right\}
  \ge 1-\alpha-\gamma-\varepsilon.
\]
\end{proposition}

\noindent\emph{Proof sketch.}
Conditional on training and calibration data, the plug-in threshold is fixed, so
test coverage equals the target score CDF evaluated at that threshold.  By the
definition of the empirical quantile,
\(\widehat F_w(\widehat q_w^{\,\mathrm{plug}})\ge1-\alpha\).  The two uniform
approximation bounds then give the claimed lower bound.  The full proof is in Appendix~\ref{app:proof-coverage-degradation}.

\bigskip
\begin{remark}[The two error bounds require separate control]
\label{rem:coverage-premises}
The quantity \(\gamma\) controls calibration-sample error in the weighted
empirical score CDF; a weighted empirical-process bound is needed to make this
term explicit.  The quantity \(\varepsilon\) controls residual input mismatch for
the score-relevant family \(\{f_q:q\in\mathbb R\}\).  If one instead confirms only
the data-dependent function
\(f_{\widehat q_w^{\,\mathrm{plug}}}\), the calibration data and threshold must be
fixed before an independent target monitoring stream is used, and the argument
is applied conditionally on \(\Dcal\).  Alternatively, a prespecified class of
threshold-indexed functions may be controlled uniformly.  If the score-relevant
functions are estimated rather than known, their estimation error must also be
included in \(\varepsilon\).
\end{remark}

\cref{prop:coverage-degradation} gives the operational interpretation of
confirmed balance for conformal prediction.  The remaining score-CDF mismatch
appears directly as a coverage loss.  Thus the right balancing-function class for
weighted conformal calibration is not arbitrary; it should see the aspects of the
input distribution that affect the score distribution.

The downstream workflow is summarized in \cref{alg:weighted-cb}.  Confirmation is
used as a gate for the covariate correction, while conformal calibration uses the
confirmed weights in the usual weighted quantile construction.

\begin{algorithm}[H]
\caption{Weighted conformal calibration after confirming a covariate correction}
\label{alg:weighted-cb}
\begin{algorithmic}[1]
\Require Training data \(\Dtr\), calibration data \(\Dcal\), proposed correction
\(w\), target input stream, global-evidence level \(\alpha_{\mathrm e}\), target
balance level \(\delta_{\mathrm{bal}}\), source level \(\eta_{\mathrm{src}}\),
conformal miscoverage \(\alpha_{\mathrm{cp}}\)
\State Fit a predictive model or score function \(s(\cdot,\cdot)\) using \(\Dtr\)
\State Compute calibration scores \(S_i=s(X_i^s,Y_i^s)\)
\State Before target monitoring, fix the balancing-function class, including score-relevant functions if used
\State Run \cref{alg:better-source} at level \(\alpha_{\mathrm e}\) to obtain source-calibrated global relative evidence, if desired
\State Run \cref{alg:balance} at target level \(\delta_{\mathrm{bal}}\) and source level \(\eta_{\mathrm{src}}\) to confirm \(\varepsilon\)-balance
\If{\(\varepsilon\)-balance is not confirmed}
  \State \Return do not deploy the proposed correction
\EndIf
\For{new target input \(x\)}
  \State Form normalized weights \(p_i^w(x),p_\infty^w(x)\) and the
  test-point-dependent threshold \(\widehat q_w(x)\) at level \(1-\alpha_{\mathrm{cp}}\)
  \State Output \(C_w(x)=\{y:s(x,y)\le \widehat q_w(x)\}\)
\EndFor
\end{algorithmic}
\end{algorithm}

\section{Numerical Experiments}
\label{sec:experiments}

This section isolates five claims developed in the preceding theory.  The
experiments examine (i) source-reference validity and KL-predicted drift of the
global e-process, (ii) false-confirmation control and locality of the balance-confirmation
procedure, (iii) validity and directional power of acceptable-region
monitoring, (iv) the downstream effect of exact and partial covariate weighting
on split conformal prediction, and (v) the distinction between contracted
confirmation bands and expanded compatibility bands when weighted source
moments are estimated.  The reported stopping rates are finite-horizon operating
characteristics; the formal guarantees concern the probability of ever crossing
or ever falsely confirming.

\subsection{Source-calibrated global relative evidence}
\label{subsec:exp-better-source}

We use the Gaussian input model
\[
  P_s^X=N(0,I_d),
  \qquad
  P_t^X=N(\mu,I_d),
  \qquad
  d=5,\quad \|\mu\|=1.2.
\]
For \(a\in\mathbb R\), the proposed correction is
\[
  w_a(x)
  =
  \exp\{a\mu^\top x-\tfrac12a^2\|\mu\|^2\},
\]
which corresponds to the corrected input distribution \(N(a\mu,I_d)\).  Under \(X\sim N(b\mu,I_d)\), the expected log-growth rate is
\[
  \E_{X\sim N(b\mu,I_d)} \left[ \log w_a(X) \right]
  =
  ab\|\mu\|^2-\frac12a^2\|\mu\|^2.
\]
We monitor \(M_n=\prod_{i=1}^n w_a(X_i)\) up to \(T=300\) with
\(\alpha=0.05\), using 1000 Monte Carlo replications.  These exponential weights
are normalized analytically, so the conservative finite-source normalizer of
\cref{subsec:conservative-normalizer} is not used.  The Gaussian drift formula
used in the theory column below is derived in \cref{app:gaussian-partial}.

\begin{table}[H]
\centering
\caption{Source-calibrated global e-process performance over 1000 Monte Carlo runs.  The horizon is \(T=300\), \(\alpha=0.05\), \(d=5\), and \(\|\mu\|=1.2\).  Median stop is taken over the runs that cross \(1/\alpha\); ``--'' marks scenarios with no crossings.}
\label{tab:better-source-results}
\begin{tabular}{p{0.34\linewidth}rrrr}
\toprule
Scenario & Theory drift & Empirical drift & Crossing rate & Median stop \\
\midrule
source-reference law, exact weight & -0.720 & -0.721 & 0.024 & 4 \\
target, exact weight & 0.720 & 0.724 & 1.000 & 4 \\
target, partial weight & 0.540 & 0.540 & 1.000 & 6 \\
target, wrong-direction weight & -0.900 & -0.901 & 0.000 & -- \\
\bottomrule
\end{tabular}
\end{table}

\Cref{tab:better-source-results} agrees with the drift calculation.  Under the
source-reference law, the empirical probability of crossing by \(T=300\) is
\(0.024\), below the anytime-valid bound \(\alpha=0.05\).  The median stopping
time of \(4\) in this row is conditional on the rare runs that cross and should
not be interpreted as power under the null.  Under the target distribution, the
exact correction has empirical drift \(0.724\), close to the theoretical value
\(0.720\).  The partial correction \(a=0.5\) also has positive drift and crosses
rapidly despite leaving residual mismatch.  Thus favorable global evidence
means that the correction improves on the source in KL risk; it does not imply
\(\varepsilon\)-balance.  Conversely, the wrong-direction correction has
empirical drift \(-0.901\) and never crosses within the horizon.  This
non-crossing is not a level-controlled rejection, although sustained negative
log-growth is a useful diagnostic of a harmful correction direction.  The next
experiment shows that this diagnostic and the local balance-confirmation procedure provide
nonredundant but non-exhaustive information.

\subsection{Covariate balance confirmation}
\label{subsec:exp-balance}

We reuse the Gaussian source and target distributions from
\cref{subsec:exp-better-source} and take the coordinate functions
\(f_j(x)=x_j\) as balancing functions.  Thus \(\mu_{t,j}\) and \(\mu_{w,j}\)
are coordinate means, and a correction of strength \(a\) has discrepancy
\(|1-a|\,|\mu_j|\) on coordinate \(j\).  The first two rows below use \(d=5\)
and \(\|\mu\|=1.2\); for \(a=0.5\), at least one coordinate discrepancy exceeds
the tolerance \(0.20\).  The last two rows use \(d=20\) and a correction that
matches the target mean on only the first five coordinates, while each of the
remaining coordinates has discrepancy greater than \(0.30\).  The exact mean
vectors are provided in the accompanying code.

Because coordinate functions are unbounded under Gaussian inputs, we use a
sub-Gaussian rather than bounded-Hoeffding confidence sequence.  If each
coordinate is \(\sigma^2\)-sub-Gaussian, the union-over-time sequence
\[
  \widehat\mu_{n,j}
  \pm
  \sqrt{
  \frac{2\sigma^2\log\{m\pi^2n^2/(3\delta)\}}{n}
  }
\]
has simultaneous coverage at least \(1-\delta\) over all \(n\ge1\) and
\(j\le m\), by the same weighted union bound as
\eqref{eq:hoeffding-radius}.  The standard Gaussian coordinates have
\(\sigma^2=1\).  Variance-adaptive confidence sequences could be substituted
under their corresponding assumptions.  We set \(\delta=0.05\), monitor to
\(T=1200\), and use 500 Monte Carlo replications.

\begin{table}[H]
\centering
\caption{Covariate balance confirmation with linear Gaussian balancing functions over 500 Monte Carlo runs.  The horizon is \(T=1200\) and the confidence-sequence level is \(\delta=0.05\).}
\label{tab:balance-results}
\footnotesize
\begin{tabular}{p{0.34\linewidth}rrrr}
\toprule
Scenario & Balancing functions & \(\varepsilon\) & Confirm rate & Median stop \\
\midrule
exact correction & 5 & 0.250 & 0.908 & 912 \\
partial correction, residual discrepancy above tolerance & 5 & 0.200 & 0.000 & -- \\
weak balance, first five balancing functions & 5 & 0.300 & 1.000 & 632 \\
weak balance, all twenty balancing functions & 20 & 0.300 & 0.000 & -- \\
\bottomrule
\end{tabular}
\end{table}

\Cref{tab:balance-results} illustrates both the value and the locality of balance confirmation.  The exact correction is confirmed in \(90.8\%\) of runs by
\(T=1200\); the remaining runs simply have not accumulated enough precision by
the finite horizon.  For the partial correction \(a=0.5\), confirmation never
occurs at tolerance \(\varepsilon=0.20\), even though
\cref{tab:better-source-results} shows favorable global evidence.  Global
improvement is therefore distinct from closeness to the target.  The final two
rows isolate locality.  A correction that balances only the first five
coordinates is confirmed when those five coordinates constitute the entire
balancing-function class, but it is never confirmed when all twenty coordinates
are monitored.  Hence the resulting certificate can only speak to the discrepancies
encoded by the prespecified class.

The global e-process also misses this particular omission.  Write the
weak-balance correction as \(P_w^X=N(\theta,I_{20})\), where \(\theta\) agrees
with the target mean \(\mu\) on the first five coordinates and is zero on the
remaining fifteen.  Then \cref{prop:kl-drift} gives
\[
  \E_t \left[ \log w(X) \right]
  =
  \frac12\|\mu\|^2-\frac12\|\mu-\theta\|^2
  =
  \frac12\sum_{j=1}^5\mu_j^2
  >0.
\]
Thus the incomplete correction still improves on the source in KL risk, and
the global process crosses almost surely under the conditions of
\cref{cor:positive-drift}.  The example exposes an overlapping blind spot:
balance confirmation cannot detect discrepancies outside the chosen class, and
the global monitor can favor an incomplete correction whenever it improves on
the source.  Detecting the omitted coordinates requires a richer
balancing-function class or other prespecified diagnostics.  Finally,
\(\Delta_{\mathcal F}(w)\le\varepsilon\) is a population property, whereas the
probability of confirmation by \(T\) and the median stopping time are
finite-horizon operating characteristics.

\subsection{Monitoring an acceptable exponential-tilt region}
\label{subsec:exp-acceptable-region}

We next illustrate the composite acceptable-region test of
\cref{subsec:exp-family}.  Let \(u\in\mathbb R^d\) satisfy \(\|u\|=1\), take
\[
  P_s^X=N(0,I_d),
  \qquad
  P_\theta^X=N(\theta u,I_d),
\]
and use the scalar feature \(\phi(x)=u^\top x\).  Then
\(\psi_s(\theta)=\theta^2/2\).  We regard
\[
  \Theta_0=[-\kappa,\kappa],
  \qquad \kappa=0.6,
\]
as the acceptable range of mean shifts.  For a scalar test direction
\(\lambda\),
\[
  \psi_{\Theta_0}(\lambda)
  =
  \kappa|\lambda|+\frac{\lambda^2}{2},
\]
so the directional e-process is
\[
  M_n(\lambda)
  =
  \exp\left\{
    \lambda\sum_{i=1}^n u^\top X_i
    -n\left(\kappa|\lambda|+\frac{\lambda^2}{2}\right)
  \right\}.
\]
To monitor departures in both signs, we use the prespecified mixture
\[
  \overline M_n
  =
  \frac12 M_n(0.4)+\frac12 M_n(-0.4).
\]
The mixture remains an e-process under every \(\theta\in\Theta_0\).  Under
\(P_\theta^X\), the expected log increment of component \(\lambda\) is
\[
  g_\theta(\lambda)
  =
  \lambda\theta-\kappa|\lambda|-\frac{\lambda^2}{2}.
\]
We monitor to \(T=300\) at level \(\alpha=0.05\) over 1000 Monte Carlo
replications.  The table reports the larger component drift
\(\max\{g_\theta(0.4),g_\theta(-0.4)\}\), not the drift of the mixture itself.
Median stopping time is conditional on crossing \(1/\alpha\).

\begin{table}[H]
\centering
\caption{Anytime-valid monitoring of the acceptable Gaussian tilt region
\(\Theta_0=[-0.6,0.6]\) over 1000 Monte Carlo runs.  The horizon is
\(T=300\), \(\alpha=0.05\), and the e-process mixes the two directions
\(\lambda=\pm0.4\).}
\label{tab:acceptable-region-results}
\begin{tabular}{p{0.36\linewidth}rrrr}
\toprule
Scenario & \(\theta\) & Best directional drift & Crossing rate & Median stop \\
\midrule
center of acceptable region & 0.0 & -0.320 & 0.000 & -- \\
boundary of acceptable region & 0.6 & -0.080 & 0.021 & 36 \\
slightly beyond acceptable region & 0.7 & -0.040 & 0.134 & 58 \\
clear positive departure & 1.0 & 0.080 & 1.000 & 37 \\
clear negative departure & -1.0 & 0.080 & 1.000 & 40 \\
\bottomrule
\end{tabular}
\end{table}

\Cref{tab:acceptable-region-results} illustrates validity and directional
power.  At the center and boundary of the acceptable region, the finite-horizon
crossing rates remain below \(0.05\), as required by the time-uniform null
guarantee.  Clear departures in either direction are detected rapidly.  The
case \(\theta=0.7\) lies outside the acceptable region, but both monitored
components have negative expected log-growth, so the crossing rate by \(T=300\)
is only \(0.134\).  This is a power limitation, not a validity failure.  A
crossing is valid evidence against the acceptable family, whereas absence of a
crossing does not establish membership in that family.  The prespecified
direction grid or mixing distribution should therefore target departures that
would materially affect the downstream decision.

\subsection{Downstream conformal effect}
\label{subsec:exp-conformal}

The next experiment illustrates the downstream consequences of exact and
partial covariate weighting in split conformal prediction.  It is a separate
deployment example rather than a formal implication of the preceding balance
experiment.  We use
\(d=5\), source inputs \(N(0,I_d)\), target inputs \(N(\mu,I_d)\) with
\(\mu=(0.8,0,\ldots,0)\), and responses
\[
  Y=0.5X_1+0.8X_1^2+0.5X_2+\epsilon,
  \qquad
  \epsilon\sim N(0,0.6^2).
\]
The fitted predictive model is deliberately misspecified: a ridge linear model
is trained on source data and the conformal score is the absolute residual.  Each
replication uses 600 training points, 800 calibration points, and 2000 target
test points.  Exact weighted calibration uses the test-point-dependent quantile
of \cref{subsec:weighted-split}; partial weighted calibration uses the Gaussian
tilt with strength \(a=0.5\); and weighted fitting minimizes the ridge objective
with the exact covariate weights before exact weighted calibration.  We report
averages over 300 replications at nominal \(90\%\) coverage.

\begin{table}[H]
\centering
\caption{Downstream split conformal performance over 300 Monte Carlo runs.  The target level is \(90\%\) coverage.  The regression model is deliberately misspecified.}
\label{tab:conformal-results}
\begin{tabular}{p{0.42\linewidth}rrr}
\toprule
Method & Coverage & Average width & Infinite-threshold rate \\
\midrule
unweighted & 0.814 & 3.655 & 0.000 \\
exact weighted calibration & 0.902 & 5.811 & 0.000 \\
partial weighted calibration & 0.838 & 4.061 & 0.000 \\
weighted fitting + exact calibration & 0.903 & 3.771 & 0.000 \\
\bottomrule
\end{tabular}
\end{table}

\Cref{tab:conformal-results} shows that unweighted conformal calibration
under-covers on the target distribution.  Exact covariate weighting restores
coverage to approximately the nominal level, although exact weighted calibration
alone produces wider intervals.  The partial correction increases average
width from \(3.655\) to \(4.061\), improves coverage only modestly from \(0.814\)
to \(0.838\), and still under-covers.  This behavior is consistent with the
general warning that an incomplete correction need not repair downstream
validity, but the experiment is illustrative rather than a theorem linking the
chosen balance tolerance to coverage.  Weighted fitting followed by exact
weighted calibration retains near-nominal coverage while reducing average width
to \(3.771\) in this misspecified example.  Here weighted calibration supplies
the covariate-shift validity repair, while weighted fitting improves efficiency.

\subsection{Finite source moments: confirmation versus compatibility band}
\label{subsec:exp-finite-source}

The preceding balance experiment used population weighted moments.  We now
estimate them from a finite source sample, which requires the band adjustment of
\cref{subsec:source-moment-uncertainty}.  We take
\(P_s^X=N(0,I_5)\), \(P_t^X=N(\mu,I_5)\) with \(\mu=(1,0,0,0,0)\), and a Gaussian
tilt correction \(P_w^X=N(\theta,I_5)\) with \(\theta=(0.65,0,0,0,0)\), so the
true discrepancy on the first coordinate is
\(|\mu_1-\theta_1|=0.35\), strictly above the tolerance \(\varepsilon=0.25\).  The
correction is therefore \emph{out of tolerance}, and the correct decision is not to
confirm.  Weighted source moments \(\mu_{w,j}\) are estimated by self-normalized
importance weighting from \(n_s=200\) source points.  We form a simultaneous
normal-approximation interval \([\ell_{w,j},u_{w,j}]\) using self-normalized
standard errors and a Bonferroni normal quantile at nominal level \(1-\eta\).
The target side uses the sub-Gaussian confidence sequence above at level
\(1-\delta\).  We set \(\eta=0.10\), \(\delta=0.05\), \(T=1500\), and use 4000
replications.  The comparison is between the compatibility band
\([\ell_{w,j}-\varepsilon,u_{w,j}+\varepsilon]\) and the confirmation band
\([u_{w,j}-\varepsilon,\ell_{w,j}+\varepsilon]\).

\begin{table}[H]
\centering
\caption{Finite-source covariate balance confirmation over 4000 Monte Carlo runs.  The correction is
out of tolerance by construction (true discrepancy \(0.35>\varepsilon=0.25\)), so any
confirmation is false.  Source moments are estimated from \(n_s=200\) points
using a normal-approximation source interval with nominal \(\eta=0.10\);
the target confidence sequence uses \(\delta=0.05\).  The formal
\(\delta+\eta\) guarantee applies when the source interval has the stated
coverage.}
\label{tab:finite-source-results}
\begin{tabular}{@{}p{0.52\linewidth}
  >{\centering\arraybackslash}p{0.24\linewidth}
  >{\centering\arraybackslash}p{0.16\linewidth}@{}}
\toprule
Decision band & Stopping rate under imbalance & Empty-band rate \\
\midrule
compatibility band \([\ell_{w,j}-\varepsilon,\,u_{w,j}+\varepsilon]\) & 0.376 & -- \\
confirmation band \([u_{w,j}-\varepsilon,\,\ell_{w,j}+\varepsilon]\) & 0.000 & 0.297 \\
\bottomrule
\end{tabular}
\end{table}

\Cref{tab:finite-source-results} makes the geometric distinction concrete.
If the expanded compatibility band were incorrectly used for confirmation, it
would falsely confirm the out-of-tolerance correction in \(37.6\%\) of runs.
Because expansion absorbs source uncertainty in the wrong direction, a target
confidence sequence can fit inside the expanded band even when the true
discrepancy exceeds the tolerance.  The contracted confirmation band has no
false confirmations in these simulations and is empty in \(29.7\%\) of runs,
thereby reporting that the desired tolerance cannot be confirmed from the
available source sample.  The formal \(\delta+\eta\) result requires the source
interval to have its stated simultaneous coverage; this experiment illustrates
the band geometry and does not itself validate the normal approximation.
Overall, the experiments support the distinct interpretations of the global
monitor, the local balance-confirmation procedure, and the acceptable-region test, while
also displaying their finite-horizon and specification-dependent limitations.

\section{Conclusion}
\label{sec:conclusion}

This paper separates the construction of a covariate correction from the statistical confirmation of its adequacy. Its primary result is an anytime-valid procedure for confirming $\varepsilon$-covariate balance for a prespecified correction, balancing-function class, and collection of tolerances. The procedure may be monitored continuously and stopped once every target confidence sequence is contained in its corresponding tolerance band. If the correction is out of tolerance for at least one balancing function, the probability of ever falsely confirming balance is at most $\delta$. When the weighted source moments are estimated, contracted confirmation bands account for source-side uncertainty and yield an analogous false-confirmation bound of $\delta+\eta$. The resulting certificate is absolute with respect to the chosen balance criteria, although necessarily local to those criteria.

The source-calibrated global e-process provides different but non-exhaustive information. Its KL-drift identity describes whether the proposed correction is globally favored over the uncorrected source distribution along the observed target stream, but crossing its evidence threshold does not confirm closeness to the target distribution. In particular, an incomplete correction may satisfy a restricted balancing-function class and simultaneously have positive e-process drift, so the blind spots of the two assessments can overlap. The global monitor is most informative for revealing globally harmful correction directions, whereas the balance procedure provides the downstream-adequacy statement associated with the prespecified functions and tolerances. The acceptable-region exponential-tilt construction extends the same anytime-valid perspective to testing whether the target stream lies beyond a prespecified family of plausible corrections. Once a correction has been balance-confirmed for the relevant criteria, it may also be deployed in weighted conformal prediction.

The guarantees depend on prespecifying the correction or acceptable region, monitoring protocol, balancing functions, tolerances, and test directions. A finite balancing-function class cannot establish equality of full input distributions, a fixed exponential-tilt direction need not detect every alternative outside the acceptable region, and finite source samples or highly variable weights can make confirmation slow or impossible. These limitations motivate methods for adaptively enriching balancing functions and test directions while preserving anytime validity, sharper treatment of source-side and weight-estimation uncertainty, and confirmation criteria tailored directly to downstream prediction, selection, and decision-making under covariate shift.

\bibliographystyle{apalike}
\bibliography{sjc}

\begin{thebibliography}{}

\bibitem[Angelopoulos and Bates, 2023]{AngelopoulosAN2023ftml}
Angelopoulos, A.~N. and Bates, S. (2023).
\newblock Conformal prediction: A gentle introduction.
\newblock {\em Foundations and Trends$^\circledR$ in Machine Learning},
  16(4):494--591.

\bibitem[Athey et~al., 2018]{AtheyS2018jrsssb}
Athey, S., Imbens, G.~W., and Wager, S. (2018).
\newblock Approximate residual balancing: Debiased inference of average
  treatment effects in high dimensions.
\newblock {\em Journal of the Royal Statistical Society Series B},
  80(4):597--623.

\bibitem[Barber et~al., 2023]{BarberRF2023aos}
Barber, R.~F., Cand{\`e}s, E.~J., Ramdas, A., and Tibshirani, R.~J. (2023).
\newblock Conformal prediction beyond exchangeability.
\newblock {\em The Annals of Statistics}, 51(2):816--845.

\bibitem[Ben-Michael et~al., 2021]{Ben-MichaelE2021arxiv}
Ben-Michael, E., Feller, A., Hirshberg, D.~A., and Zubizarreta, J.~R. (2021).
\newblock The balancing act in causal inference.
\newblock arXiv:2110.14831.

\bibitem[Chan et~al., 2016]{ChanKCG2016jrsssb}
Chan, K. C.~G., Yam, S. C.~P., and Zhang, Z. (2016).
\newblock Globally efficient non-parametric inference of average treatment
  effects by empirical balancing calibration weighting.
\newblock {\em Journal of the Royal Statistical Society Series B},
  78(3):673--700.

\bibitem[Choi, 2026]{Choi2026testing}
Choi, S. (2026).
\newblock Anytime-valid confirmation of label-shift corrections.
\newblock In {\em ICML 2026 Workshop on Hypothesis Testing}.

\bibitem[Cortes et~al., 2010]{CortesC2010neurips}
Cortes, C., Mansour, Y., and Mohri, M. (2010).
\newblock Learning bounds for importance weighting.
\newblock In {\em Advances in Neural Information Processing Systems (NeurIPS)}.

\bibitem[Gretton et~al., 2012]{GrettonA2012jmlr}
Gretton, A., Borgwardt, K.~M., Rasch, M.~J., and Sch{\"o}lkopf, B. (2012).
\newblock A kernel two-sample test.
\newblock {\em Journal of Machine Learning Research}, 13:723--773.

\bibitem[Gr{\"u}nwald et~al., 2024]{GrunwaldP2024jrsssb}
Gr{\"u}nwald, P., de~Heide, R., and Koolen, W. (2024).
\newblock Safe testing.
\newblock {\em Journal of the Royal Statistical Society Series B},
  86(5):1091--1128.

\bibitem[Hainmueller, 2012]{HainmuellerJ2012pa}
Hainmueller, J. (2012).
\newblock Entropy balancing for causal effects: A multivariate reweighting
  method to produce balanced samples in observational studies.
\newblock {\em Political Analysis}, 20(1):25--46.

\bibitem[Howard et~al., 2021]{HowardSR2021aos}
Howard, S.~R., Ramdas, A., McAuliffe, J., and Sekhon, J. (2021).
\newblock Time-uniform, nonparametric, nonasymptotic confidence sequences.
\newblock {\em The Annals of Statistics}, 49(2):1055--1080.

\bibitem[Imai and Ratkovic, 2014]{ImaiK2014jrsssb}
Imai, K. and Ratkovic, M. (2014).
\newblock Covariate balancing propensity score.
\newblock {\em Journal of the Royal Statistical Society Series B},
  76(1):243--263.

\bibitem[Kanamori et~al., 2009]{KanamoriT2009jmlr}
Kanamori, T., Hido, S., and Sugiyama, M. (2009).
\newblock A least-squares approach to direct importance estimation.
\newblock {\em Journal of Machine Learning Research}, 10:1391--1445.

\bibitem[Manole and Ramdas, 2023]{ManoleT2023ieeetit}
Manole, T. and Ramdas, A. (2023).
\newblock Martingale methods for sequential estimation of convex functionals
  and divergences.
\newblock {\em IEEE Transactions on Information Theory}, 69(7):4641--4658.

\bibitem[M{\"u}ller, 1997]{MullerA1997aap}
M{\"u}ller, A. (1997).
\newblock Integral probability metrics and their generating classes of
  functions.
\newblock {\em Advances in Applied Probability}, 29(2):429--443.

\bibitem[Papadopoulos et~al., 2002]{PapadopoulosH2002ecml}
Papadopoulos, H., Proedrou, K., Vovk, V., and Gammerman, A. (2002).
\newblock Inductive confidence machines for regression.
\newblock In {\em Proceedings of the European Conference on Machine Learning
  (ECML)}.

\bibitem[Ramdas et~al., 2023]{RamdasA2023ss}
Ramdas, A., Gr{\"u}nwald, P., Vovk, V., and Shafer, G. (2023).
\newblock Game-theoretic statistics and safe anytime-valid inference.
\newblock {\em Statistical Science}, 38(4):576--601.

\bibitem[Schuirmann, 1987]{SchuirmannDJ1987jpb}
Schuirmann, D.~J. (1987).
\newblock A comparison of the two one-sided tests procedure and the power
  approach for assessing the equivalence of average bioavailability.
\newblock {\em Journal of Pharmacokinetics and Biopharmaceutics}, 15:657--680.

\bibitem[Shafer and Vovk, 2008]{ShaferG2008jmlr}
Shafer, G. and Vovk, V. (2008).
\newblock A tutorial on conformal prediction.
\newblock {\em Journal of Machine Learning Research}, 9:371--421.

\bibitem[Shekhar and Ramdas, 2024]{ShekharS2024ieeetit}
Shekhar, S. and Ramdas, A. (2024).
\newblock Nonparametric two-sample testing by betting.
\newblock {\em IEEE Transactions on Information Theory}, 70(2):1178--1203.

\bibitem[Shimodaira, 2000]{ShimodairaH2000jspi}
Shimodaira, H. (2000).
\newblock Improving predictive inference under covariate shift by weighting the
  log-likelihood function.
\newblock {\em Journal of Statistical Planning and Inference}, 90:227--244.

\bibitem[Sriperumbudur et~al., 2012]{SriperumbudurBK2012ejs}
Sriperumbudur, B.~K., Fukumizu, K., Gretton, A., Sch{\"o}lkopf, B., and
  Lanckriet, G. R.~G. (2012).
\newblock On the empirical estimation of integral probability metrics.
\newblock {\em Electronic Journal of Statistics}, 6:1550--1599.

\bibitem[Sugiyama et~al., 2007]{SugiyamaM2007neurips}
Sugiyama, M., Nakajima, S., Kashima, H., von B{\"u}nau, P., and Kawanabe, M.
  (2007).
\newblock Direct importance estimation with model selection and its application
  to covariate shift adaptation.
\newblock In {\em Advances in Neural Information Processing Systems (NeurIPS)}.

\bibitem[Sugiyama et~al., 2012]{SugiyamaM2012book}
Sugiyama, M., Suzuki, T., and Kanamori, T. (2012).
\newblock {\em Density Ratio Estimation in Machine Learning}.
\newblock Cambridge University Press.

\bibitem[Tibshirani et~al., 2019]{TibshiraniR2019neurips}
Tibshirani, R.~J., Barber, R.~F., Cand{\`e}s, E.~J., and Ramdas, A. (2019).
\newblock Conformal prediction under covariate shift.
\newblock In {\em Advances in Neural Information Processing Systems (NeurIPS)}.

\bibitem[Ville, 1939]{VilleJ1939phd}
Ville, J. (1939).
\newblock {\em {\'E}tude Critique de la Notion de Collectif}.
\newblock PhD thesis, Universit{\'e} de Paris.

\bibitem[Vovk et~al., 2005]{VovkV2005book}
Vovk, V., Gammerman, A., and Shafer, G. (2005).
\newblock {\em Algorithmic Learning in a Random World}.
\newblock Springer.

\bibitem[Wald, 1945]{WaldA1945aoms}
Wald, A. (1945).
\newblock Sequential tests of statistical hypotheses.
\newblock {\em The Annals of Mathematical Statistics}, 16(2):117--186.

\bibitem[Waudby-Smith and Ramdas, 2024]{Waudby-SmithI2024jrsssb}
Waudby-Smith, I. and Ramdas, A. (2024).
\newblock Estimating means of bounded random variables by betting.
\newblock {\em Journal of the Royal Statistical Society Series B}, 86:1--27.

\bibitem[Wellek, 2010]{WellekS2010book}
Wellek, S. (2010).
\newblock {\em Testing Statistical Hypotheses of Equivalence and
  Noninferiorit}.
\newblock Chapman \& Hall/CRC, 2nd edition.

\bibitem[Zubizarreta, 2015]{ZubizarretaJR2015jasa}
Zubizarreta, J.~R. (2015).
\newblock Stable weights that balance covariates for estimation with incomplete
  outcome data.
\newblock {\em Journal of the American Statistical Association},
  110(511):910--922.

\end{thebibliography}

\clearpage
\appendix

\section{Proofs}
\label{app:proofs}

\subsection{Proof of Theorem~\ref{thm:input-eprocess}}
\label{app:proof-input-eprocess}

\begin{proof}
Fix a realization of \(\cG_0\).  By \cref{ass:prespecified}, \(w\) is then fixed
and satisfies \(\E_s \left[ w(X) \right]=1\).  Under \(H_0^{\mathrm{src}}\), conditionally on
the past,
\[
  \E[w(X_i)\mid\cG_{i-1}]=\E_s \left[ w(X) \right]=1.
\]
Therefore \(M_n\) is a nonnegative martingale under the corresponding conditional
law.  Ville's inequality gives the conditional time-uniform bound, and averaging
over \(\cG_0\) gives \eqref{eq:ville-input}.
\end{proof}

\subsection{Proof of Proposition~\ref{prop:kl-drift}}
\label{app:proof-kl-drift}

\begin{proof}
Since \(dP_w^X/dP_s^X=w\) and \(Q^X\ll P_w^X\ll P_s^X\), the chain rule for
Radon--Nikodym derivatives gives
\[
\log\frac{dQ^X}{dP_s^X}
=
\log\frac{dQ^X}{dP_w^X}
+
\log\frac{dP_w^X}{dP_s^X}
=
\log\frac{dQ^X}{dP_w^X}
+
\log w
\qquad Q^X\text{-a.s.}
\]
By hypothesis,
\[
  \KL(Q^X\|P_s^X)
  =
  \E_Q \left[ \log\frac{dQ^X}{dP_s^X} \right]
  <\infty,
  \qquad
  \E_Q \left[ |\log w(X)| \right]<\infty.
\]
It follows that
\[
  \E_Q \left[ \log\frac{dQ^X}{dP_w^X} \right]
  =
  \KL(Q^X\|P_w^X)
  <\infty,
\]
so each term above is integrable.  Taking expectations under \(Q^X\) and
rearranging,
\[
\E_Q \left[ \log w(X) \right]
=
\E_Q \left[ \log\frac{dQ^X}{dP_s^X} \right]
-
\E_Q \left[ \log\frac{dQ^X}{dP_w^X} \right]
=
\KL \left(Q^X\|P_s^X \right)-\KL \left(Q^X\|P_w^X \right).
\]
\end{proof}

\subsection{Proof of Corollary~\ref{cor:positive-drift}}
\label{app:proof-positive-drift}

\begin{proof}
Since \(\E_Q \left[ |\log w(X)| \right]<\infty\), the strong law of large numbers gives
\[
  \frac{1}{n}\log M_n
  =
  \frac{1}{n}\sum_{i=1}^n \log w(X_i)
  \to
  \E_Q \left[ \log w(X) \right]
  \quad\text{almost surely,}
\]
and by \cref{prop:kl-drift} the limit equals
\(\KL(Q^X\|P_s^X)-\KL(Q^X\|P_w^X)\), which is strictly positive by hypothesis.
Hence \(\log M_n\to+\infty\) almost surely, so \(M_n\ge 1/\alpha\) for all
sufficiently large \(n\) and \(\tau_{\conf}<\infty\) almost surely.
\end{proof}

\subsection{Proof of Theorem~\ref{thm:balance-confirmation}}
\label{app:proof-balance-confirmation}

\begin{proof}
On the simultaneous coverage event in \eqref{eq:simultaneous-cs}, every confidence sequence contains its
true mean for all monitoring indices.  If
\(|\mu_{t,j}-\mu_{w,j}|>\varepsilon_j\) for some \(j\), then no interval
containing \(\mu_{t,j}\) can be a subset of
\([\mu_{w,j}-\varepsilon_j,\mu_{w,j}+\varepsilon_j]\).  Hence
\(\tau_{\mathrm{bal}}=\infty\) on the simultaneous coverage event.  The false-confirmation probability is at most the failure probability \(\delta\).
\end{proof}

\subsection{Proof of Proposition~\ref{prop:conservative-normalizer}}
\label{app:proof-conservative-normalizer}

\begin{proof}
Because \(U_Z\) and \(Z\) are \(\cG_0\)-measurable, fix a realization of
\(\cG_0\) for which \(Z\le U_Z\).  Under the source-reference calibration null,
\(X_n\mid\cG_{n-1}\sim P_s^X\), and therefore
\[
  \E\left[
    \left.
    \frac{\widetilde w(X_n)}{U_Z}
    \right|\cG_{n-1}
  \right]
  =
  \frac{Z}{U_Z}
  \le1.
\]
Thus, for every such realization of \(\cG_0\), \((M_n)\) is a nonnegative
supermartingale under the conditional law, and Ville's inequality gives crossing
probability at most \(\alpha\).  Averaging over \(\cG_0\) and adding the failure
probability \(\Prob\left\{Z>U_Z\right\}\le\eta\) yields the stated \(\alpha+\eta\) bound.
\end{proof}

\subsection{Proof of Corollary~\ref{cor:finite-source-balance}}
\label{app:proof-finite-source-balance}

\begin{proof}
Let \(E_s\) be the event that all source intervals cover and let \(E_t\) be the
time-uniform target coverage event in \eqref{eq:simultaneous-cs}.  On
\(E_s\cap E_t\), for every \(j\) and \(n\),
\[
  B_j^{\mathrm{conf}}
  =
  \bigcap_{\nu\in[\ell_{w,j},u_{w,j}]}
  [\nu-\varepsilon_j,\nu+\varepsilon_j]
  \subseteq
  [\mu_{w,j}-\varepsilon_j,\mu_{w,j}+\varepsilon_j].
\]
If \(|\mu_{t,j}-\mu_{w,j}|>\varepsilon_j\) for some \(j\), then no interval
containing \(\mu_{t,j}\) can be contained in \(B_j^{\mathrm{conf}}\).  Hence
\(\tau_{\mathrm{bal}}^{\mathrm{fs}}=\infty\) on \(E_s\cap E_t\).  Therefore
\[
  \Prob\left\{\tau_{\mathrm{bal}}^{\mathrm{fs}}<\infty\right\}
  \le
  \Prob\left\{E_s^c\right\}+\Prob\left\{E_t^c\right\}
  \le\eta+\delta.
\]
\end{proof}

\subsection{Proof of Proposition~\ref{prop:coverage-degradation}}
\label{app:proof-coverage-degradation}

\begin{proof}
If the total calibration weight is zero, the stated convention gives
\(\widehat q_w^{\,\mathrm{plug}}=+\infty\), so the result is immediate.
Otherwise, \(\widehat q_w^{\,\mathrm{plug}}\) is fixed given
\((\Dtr,\Dcal)\) and does not depend on the new test input \(X\).  Hence
\[
\Prob_t\left\{Y\in C_w^{\mathrm{plug}}(X)\mid\Dtr,\Dcal\right\}
=
F_t(\widehat q_w^{\,\mathrm{plug}}).
\]
By definition of the empirical weighted quantile,
\[
  \widehat F_w(\widehat q_w^{\,\mathrm{plug}})\ge1-\alpha.
\]
The assumed uniform bounds therefore imply
\begin{align*}
  F_t(\widehat q_w^{\,\mathrm{plug}})
  &\ge
  F_w(\widehat q_w^{\,\mathrm{plug}})-\varepsilon \\
  &\ge
  \widehat F_w(\widehat q_w^{\,\mathrm{plug}})-\gamma-\varepsilon \\
  &\ge
  1-\alpha-\gamma-\varepsilon.
\end{align*}
\end{proof}

\section{Additional Derivations}
\label{app:derivations}

\subsection{Gaussian drift for partial correction}
\label{app:gaussian-partial}

Let \(P_s^X=N(0,I_d)\), \(P_t^X=N(\mu,I_d)\), and propose
\[
  P_w^X=N(a\mu,I_d).
\]
Then
\[
  \KL(P_t^X\|P_s^X)=\frac12\|\mu\|^2,
\]
and
\[
  \KL(P_t^X\|P_w^X)=\frac12\|(1-a)\mu\|^2.
\]
Therefore
\[
  \E_t \left[ \log w(X) \right]
  =
  \frac12\left(1-(1-a)^2\right)\|\mu\|^2.
\]
This is positive for \(0<a<2\).  Hence a partial correction can receive positive global relative evidence even
if it is not exact.  Covariate balance confirmation is needed to determine
whether the residual mismatch is tolerable.

\subsection{Weighted conformal score CDF under covariate shift}
\label{app:score-cdf}

Let \(S=s(X,Y)\).  Under covariate shift,
\[
  P_t(Y\mid X)=P_s(Y\mid X).
\]
For any threshold \(q\),
\[
  F_t(q)
  =
  \E_t\left[\Prob_s\left\{s(X,Y)\le q\mid X\right\}\right].
\]
If \(w=w^\star\), then
\[
  F_t(q)
  =
  \frac{\E_s\left[w(X)\Prob_s\left\{s(X,Y)\le q\mid X\right\}\right]}
       {\E_s \left[ w(X) \right]}
  =
  F_w(q).
\]
Thus exact covariate weighting repairs the score distribution used by conformal
calibration.

\end{document}